\documentclass[12pt]{article}
\usepackage{blindtext}
\usepackage{float}
\usepackage{mathrsfs} 
\usepackage{PRIMEarxiv}
\usepackage{pdflscape}   % in your preamble
\usepackage{booktabs}    % for \toprule, \midrule, etc.
\usepackage[utf8]{inputenc} % allow utf-8 input
\usepackage[T1]{fontenc}    % use 8-bit T1 fonts
\usepackage{hyperref}       % hyperlinks
\usepackage{url}    
\usepackage{amsmath}
\usepackage{graphicx,psfrag,epsf}
\usepackage{enumerate}
\usepackage{natbib}
\usepackage{amsmath}
\usepackage{amssymb}
\usepackage{url}
\usepackage{hyperref}
\usepackage{amsmath, amssymb, amsthm}
\usepackage{algorithm}
\usepackage{algpseudocode}
\usepackage{graphicx}
\usepackage{caption}
\usepackage{subcaption}
\usepackage{graphicx}
\usepackage{comment}
\usepackage{minted}
\usepackage{booktabs}
  
\usepackage{tikz}
\usetikzlibrary{shapes.geometric, arrows}

\theoremstyle{plain}
\newtheorem{theorem}{Theorem}[section]

\newtheorem{proposition}{Proposition}[section]
\theoremstyle{definition}
\newtheorem{corollary}{Corollary}

\theoremstyle{remark}
\newtheorem{remark}{Remark}[section]

\graphicspath{{media/}}     % organize your images and other figures under media/ folder

\begin{document}

\def\spacingset#1{\renewcommand{\baselinestretch}%
{#1}\normalsize} \spacingset{1}

%%%%%%%%%%%%%%%%%%%%%%%%%%%%%%%%%%%%%%%%%%%%%%%%%%%%%%%%%%%%%%%%%%%%%%%%%%%%%%

\title{\bf Efficient and Fast Sampling from Arbitrary Probability Kernels using Sliced Gibbs Sampler}

\author{
  Prithwish Ghosh\thanks{Department of Statistics, North Carolina State University, Raleigh, NC, USA.} 
  \and
 \bf Sujit K. Ghosh\thanks{Department of Statistics, North Carolina State University, Raleigh, NC, USA.}
}

 \maketitle

%\if0\blind
%{
 % \bigskip
  %\bigskip
%  \bigskip
 % \begin{center}
  %  {\LARGE\bf Automated Sliced Gibbs Sampling for Non-Standard Kernels}
%\end{center}
 % \medskip
%} \fi

\begin{abstract}
An Automated Sliced Gibbs (ASG) framework is proposed for fully automated Markov chain Monte Carlo (MCMC) sampling from arbitrary finite-dimensional probability kernels. The method is designed for unnormalized densities that may be
non-smooth, heavy-tailed, highly multimodal, or supported on
irregular regions, subject to the regularity conditions stated
below. A Cauchy transformation–based effective support estimator is combined with slice-driven Gibbs updates. This construction removes the need for user-specified truncation bounds, proposal scales, step-size tuning, or conditional optimization. Unlike existing slice samplers, ASG does not require manually chosen bracket widths or geometric insight into the support. All calibration is performed automatically within each Gibbs cycle.
The resulting Markov chain leaves the target exactly invariant and is uniformly ergodic under stated regularity conditions. Automated support detection allows efficient movement across disconnected high-density regions. The sampler adapts to sharp curvature and irregular geometry without gradient information.
Extensive numerical experiments evaluate performance on complex kernels, including univariate Beta mixtures, multivariate Rosenbrock and Ackley benchmarks, and non-smooth kernels derived from generalized LASSO-type loss functions. Across kernels exhibiting multimodality
or disconnected high-density regions, ASG consistently achieves higher effective sample size than Random
   Walk Metropolis–Hastings and adaptive Gibbs variants; on smoothly curved,
   unimodal kernels, ASG remains competitive with existing slice-based methods but
can be outperformed by proposals  such as elliptical slice sampling whose geometry is well matched to the target's curvature. The framework provides a scalable and general-purpose solution for sampling from complicated probability kernels where existing algorithms require substantial tuning or exhibit slow mixing.

\end{abstract}

\noindent%

\section{Introduction}
\label{sec:intro}
Markov chain Monte Carlo (MCMC) methods have become indispensable tools in modern statistical inference,
enabling researchers to approximate complex posterior distributions that are analytically intractable
\cite{gelfand1990sampling, robert1999monte, brooks2011handbook}.
Since the seminal works of \cite{metropolis1953equation} and \cite{hastings1970monte},
MCMC algorithms have revolutionized Bayesian computation, providing general-purpose procedures
for sampling from high-dimensional target distributions.\\
Among these methods, the Gibbs sampler, introduced in the image analysis context by
\cite{geman1984stochastic} and formally popularized in statistics by
\cite{gelfand1990sampling} and \cite{casella1992explaining},
remains a cornerstone due to its conceptual simplicity and coordinate-wise updating mechanism.
By decomposing a complex joint distribution into a sequence of univariate conditionals,
the Gibbs sampler avoids the need for explicit proposal tuning.
This structure has motivated a vast literature on theoretical properties
\cite{tierney1994markov, roberts1997geometric, cowles1998simulation, roberts2004general, jones2001honest},
as well as algorithmic extensions for non-conjugate and high-dimensional problems
\cite{liu1994covariance, gamerman2006markov, liu2001monte, marshall2009ergodicity}. Despite these advances, classical Gibbs sampling suffers from several well-known limitations.
First, it requires the full conditional distributions to be available in a standard distributional form that are easy to sample from using unnormalized kernels.
For arbitrary kernels—especially unnormalized, multimodal, or geometrically irregular densities—
this method seldom works as it becomes prohibitive to sample easily from the full conditional density kernels, leading to a bottleneck. Moreover, as the Gibbs sampler moves coordinate wise (even when used as a block), it can get trapped in low or high density regions and takes a long time to escape from local modes, making the algorithm mixing very slowly.
The Rosenbrock or ``banana-shaped'' distributions, widely used for MCMC benchmarking
\cite{pagani2022n, haario1999adaptive, roberts2001optimal},
illustrate how naive Gibbs updates may get trapped in curved or narrow probability ridges.
Second, the necessity of manually specifying conditional support bounds severely restricts automation,
especially for high-dimensional or nonstandard target densities.
A wide range of methodological innovations has aimed to mitigate these deficiencies.
Adaptive MCMC methods dynamically adjust proposal scales or covariance structures
during sampling \cite{andrieu2008tutorial, roberts2009examples, neal2012optimal}.
Notably, the Adaptive Metropolis (AM) algorithm of \cite{haario2001adaptive} and subsequent
geometrically adaptive variants \cite{andrieu2008tutorial, roberts2009examples}
achieved considerable success in improving chain mixing,
though they still rely on user-specified tuning heuristics and require differentiability of the log-target.
More recently, gradient-based samplers such as Hamiltonian Monte Carlo (HMC)
\cite{neal2011mcmc, betancourt2017conceptual} and the No-U-Turn Sampler
have gained prominence for their efficiency in smooth continuous spaces,
but their dependence on gradient evaluations limits applicability to non-differentiable or bounded kernels. An alternative line of research centers on \emph{slice sampling},
extensively explored by \cite{neal2003slice}, which introduces auxiliary variables
to sample uniformly from the region under the density curve.
The method avoids explicit step-size tuning and adapts automatically to the target’s local scale.
Slice sampling has since been generalized to stepping-out and shrinking procedures
\cite{neal2003slice}, reflective and elliptical variants for Gaussian process models
\cite{murray2010elliptical}, and adaptive slice methods that modify the bracket width dynamically
\cite{thompson1996adaptive, tibbits2014automated}.
However, these algorithms typically require manually defined search intervals or prior geometric knowledge
of the support, limiting full automation when dealing with arbitrary unnormalized kernels.\\Several works have explored hybrid frameworks combining Gibbs and slice sampling ideas.    The closest antecedent is the auxiliary-variable Gibbs sampler of
   \cite{damlen1999gibbs}, which requires the level sets in closed form;
   ASG removes this requirement through automatic effective-support
   bracketing while preserving exact invariance (Proposition~\ref{prop:exact}).
The adaptive Gibbs sampler of \cite{liu1994covariance} and the adaptive block schemes sought to enhance Gibbs efficiency by learning local dependence structures.
Other approaches, such as random-direction slice sampling \cite{neal2003slice}
and component-wise adaptive slice updates \cite{tibbits2014automated},
partially automate step selection but still rely on fixed bounding boxes.
From a theoretical standpoint, ergodicity and convergence analyses for these samplers
have been studied extensively \cite{tierney1994markov, roberts2004general, cowles1998simulation}.
Despite these developments, a unified, theoretically valid, and \emph{fully automated} sampling
framework for arbitrary multivariate kernels remains elusive.\\
We propose an \emph{Automated Sliced Gibbs(ASG)} sampling method that overcomes the support-bound bottleneck by incorporating adaptive support estimation with univariate slice updates. The key idea is to normalize the kernel over its \emph{effective support}, automatically identified via a Cauchy-based transformation (similar to activation function in neural net), and to use {efficient} rejection-based slice sampling to ensure efficient exploration of the conditional support. \\
This article develops a general-purpose MCMC method for sampling
from a broad class of nonnegative, integrable, unnormalized kernels while requiring minimal manual specification
of support bounds.
The ASG framework makes three methodological contributions. First, it
constructs coordinate-wise effective-support intervals automatically
using a Cauchy reparameterization of the real line, thereby avoiding
the need for manually specified coordinate-wise support bounds.
Second, it combines these automatically constructed intervals with
coordinate-wise sliced Gibbs updates and a retain-and-resample
mechanism that preserves the target slice distribution even when the
effective-support interval does not contain the complete slice.
Third, it provides a unified implementation in which the numerical
support construction, slice update, and sampling diagnostics are
integrated within a single MCMC procedure. The theoretical results
separate exact invariance of the coordinate update from the stronger
regularity conditions required for geometric ergodicity.\\
For a simple 1-d illustration, Figure~\ref{com_densi} provides a comparison of the performance of the proposed ASG sampler with a standard {Hamiltonian Monte Carlo} (HMC) algorithm on a univariate mixture of scaled beta distributions. The HMC fails to adapt to the multimodality and potentially unbounded kernel, resulting in only a limited sampling from only one component of the mixture while failing to sample from regions under the two other modes of the density. In contrast, the proposed ASG sampler very quickly generates samples from the entire support of the density.\\
\begin{figure}[!ht]
    \centering
     \begin{subfigure}[b]{0.45\linewidth}
        \includegraphics[width=\linewidth]{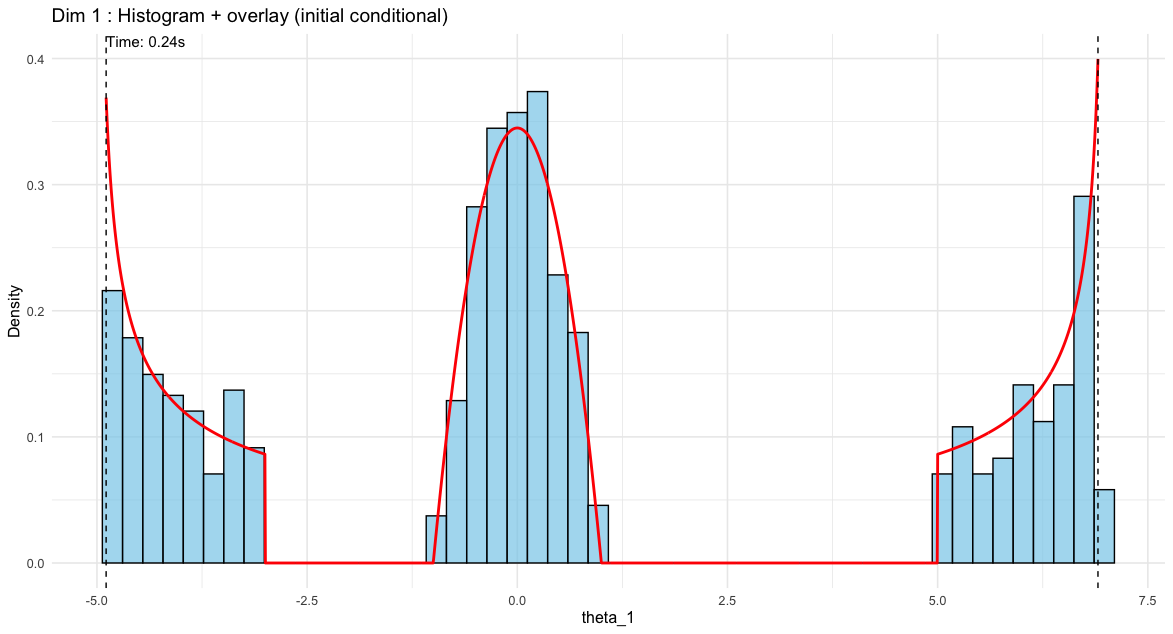}
    \end{subfigure}
     \begin{subfigure}[b]{0.45\linewidth}
        \includegraphics[width=\linewidth, height = 4cm]{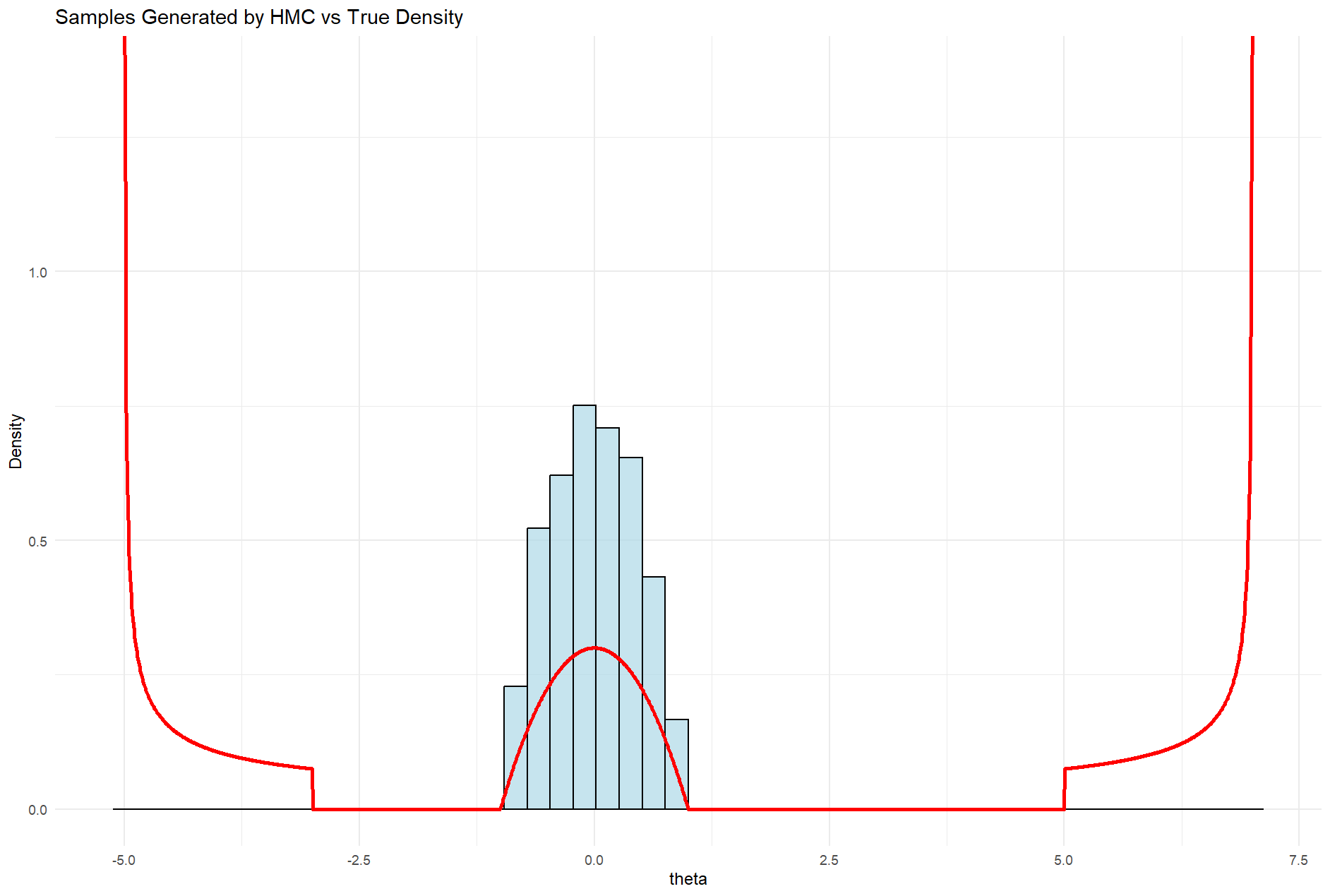}
    \end{subfigure}
\caption{Comparison of proposed ASG vs. HMC for the univariate mixture density (see Section~\ref{univariate} for details). Notice that no samples are generated from the first and third components of the mixture density by the HMC in the figure (b).}
\label{com_densi}
\end{figure}
The remainder of the paper is organized as follows. Section~\ref{sec:meth} introduces the automated slice sampling framework for unnormalized kernels. It presents the effective support estimation strategy, develops the bounded univariate slice updates, and establishes invariance and ergodicity in both univariate and multivariate Gibbs settings. Diagnostic tools and convergence results are summarized. In section~\ref{sec:station} we showed the stationarity and geometric Ergodicity of our proposed sampler(ASG), and in section~\ref{benchmark} effective sample size is introduced as the primary benchmarking metric.
Section~\ref{sec:examples} reports empirical results based on fixed-sample experiments conducted under a standardized computational environment. Performance is evaluated across a diverse collection of challenging targets, including multimodal mixtures, classical optimization benchmarks and in Section~\ref{sec:non smooth} present results based on various non-smooth loss-derived kernels. Across multimodal and disconnected-support settings, the proposed sampler consistently delivers substantially higher effective sample size per second than Random Walk Metropolis–Hastings, while maintaining comparable computational cost. Finally, general conclusions are presented in Section~\ref{sec:conclusion}.

\section{Proposed Automated Sliced Gibbs Sampler}
\label{sec:meth}
This section develops a mathematically rigorous framework for Markov chain Monte Carlo (MCMC) sampling from unnormalized multivariate density functions, with particular emphasis on kernels defined on either bounded or unbounded supports. The proposed methodology integrates three core components: effective support estimation, a bounded–proposal variant of univariate slice sampling, and multivariate Gibbs sampling. Diagnostic tools are incorporated to facilitate empirical assessment of convergence and mixing behavior. The methodology can be easily implemented in any coding platform and is both dimension-agnostic and support-agnostic. The exposition proceeds in two stages. The effective support determination is first developed for the univariate case. The construction is then extended to arbitrary dimensions through a Gibbs sampling scheme based on conditional slice updates. It is shown algorithmically below~\ref{alg:asg}

\subsection{Sliced Gibbs Sampler: Formulation and Intuition}\label{sec:generic-1d-slice}
Let $K(x)$, $x\in\mathbb{R}^m$, denote a nonnegative kernel function (with respect to a dominating measure, e.g., Lebesgue measure) satisfying
\begin{equation}
K(x)\ge 0,\qquad Z :=\int_{\mathbb{R}^m} K(x)\,dx <\infty.
\label{eq: K(x)}
\end{equation}
Our goal is to generate (possibly correlated) samples from the normalized density
\begin{equation}
p(x)= \frac{K(x)}Z,
\label{eq: p(x)}
\end{equation}
where the normalizing constant $Z$ may not available in many practical applications. Moreover, $K$ may be \emph{multimodal}, making traditional Gibbs algorithms prone to ``mode–trapping.'' The \emph{slice sampling} framework circumvents this difficulty by augmenting the state space and sampling from regions defined by the level sets of the kernel. This naturally enables \emph{large jumps across low–density valleys}, allowing the chain to move efficiently between modes.

\subsubsection*{Augmentation by a Slice Variable}
Introduce an auxiliary variable $u>0$ and consider the joint kernel
$H(u, x)=\mathbb{I}\{K(x)>u\}.$
Marginalizing out $u$, we recover the target kernel:
$\int_0^\infty H(x,u)\,du = K(x).$
Thus, if we can sample from the normalized version of $H(x,u)$, the marginal samples of $x$ will be sampled from the target density $p(x)$ given in eq.(\ref{eq: p(x)}).
\noindent The standard slice sampling strategy alternates between, drawing uniformly a slice height $u$ beneath $K(x)$, i.e., $u\sim Unif(0, K(x))$, and drawing a new state $x$ uniformly from the region where $K(x)$ exceeds this height, i.e. uniformly from $\mathcal{S}_u := \{x\in\mathbb{R}^m: K(x)>u\}$ for a $u$ sampled in step (i) above.\\
This construction makes it possible to traverse disconnected or narrow high–density regions because the algorithm samples \emph{within horizontal slices} of the density rather than following the curved geometry of the density's conditional distributions. However, it is well known the step (ii) of the above algorithm is difficult to accomplish for a general kernel function $K(x)$ in higher dimensions. In order to mitigate the problem of generating samples uniformly from complex high-dimensional set $\mathcal{S}_u$ for a given $u>0$, we propose a sliced version of the Gibbs sampler.

\subsubsection*{The Sliced Gibbs Sampler}
Let $\mathcal{S}_1(x_2, x_3, \ldots, x_m; u)=\{z\in\mathbb{R}: K(z,x_2,\ldots,x_m)>u\}$ denote level set for the $x_1$ conditional on rest of the coordinates for a given height $u>0$ and more generally for $j=1,2,\ldots$, let the level set for the $j$-th coordinate, $x_j$, conditional on rest of the coordinates $x_1,\ldots,x_{j-1},x_{j+1},\ldots, x_m$  and the height $u$ be denoted by
\begin{equation}
\mathcal{S}_j(x_1,\ldots,x_{j-1},x_{j+1},\ldots, x_m; u)=\{z\in\mathbb{R}: K(x_1,\ldots, x_{j-1}, z, x_{j+1},\ldots, x_m)>u\}.
\label{eq:levelset}
\end{equation}
Notice that $\mathcal{S}_m(x_1, x_2, \ldots, x_{m-1}; u)=\{z\in\mathbb{R}: K(x_1,x_2,\ldots,x_{m-1}, z)>u\}$ is the level set for $x_m$ conditional on the first $(m-1)$ coordinates.\\ 
Also, let $A_j(x_1,\ldots,x_{j-1},x_{j+1},\ldots,x_m; u)$ denote the Lebesgue measure of the $j$-th level set given in eq.(\ref{eq:levelset}). More specifically notice that 
\begin{equation}
A_j(x_1,\ldots,x_{j-1},x_{j+1},\ldots,x_m; u) = \int \mathbb{I}(K(x_1,\ldots, x_{j-1}, z, x_{j+1},\ldots, x_m)>u)\,dz
\label{eq:lebesgue}
\end{equation}
Let $\text{Unif}\big(S_j(x_1,\dots,x_{j-1},x_{j+1},\ldots,x_m; u)\big)$ denote the (conditional) uniform distribution over the $j$-th level set $S_j(x_1,\dots,x_{j-1},x_{j+1},\ldots,x_m; u)$ given the rest of the values $x_1,\dots,$ $x_{j-1},x_{j+1},\ldots,x_m$ and a given height $u>0$. Notice that the Lebesgue measure in eq. (\ref{eq:lebesgue}) is finite and in fact, it easily follows, $A_j(x_1,\ldots,x_{j-1},x_{j+1},\ldots,x_m; u)\leq \frac{1}{u}\int K(x_1,\ldots, x_{j-1}, z, x_{j+1},\ldots, x_m)\;dz$ for any $u>0$.

Starting with an initial value $x^{(0)}=(x_1^{(0)},\dots,x_m^{(0)})\;$ such that $K(x^{(0)})>0$, we start drawing new height and components sequentially for $t=1,2,\ldots$, as follows:
\begin{equation}
\begin{aligned}
u^{(t)} &\sim \text{Unif}\big(0,K(x^{(t-1)})\big)\propto \mathbb{I}(K(x^{(t-1)})>u\}.\\
x_1^{(t)} &\sim \text{Unif}\big(S_1(x_2^{(t-1)},\dots,x_m^{(t-1)}; u^{(t)}\big) \propto \mathbb{I}\{K(x_1,x_2^{(t-1)},\dots,x_m^{(t-1)}) > u^{(t)}\},\\
&\ \vdots\\
x_j^{(t)} &\sim \text{Unif}\big(S_j(x_1^{(t)},\dots,x_{j-1}^{(t)},x_{j+1}^{(t-1)},\ldots,x_m^{(t-1)}; u^{(t)}\big)\\ 
& \hspace{1.5in} \propto \mathbb{I}\{K(x_1^{(t)},\dots,x_{j-1}^{(t)}, x_j, x_{j+1}^{(t-1)},\ldots,x_m^{(t-1)} > u^{(t)}\},\\
& \vdots\\
x_m^{(t)} &\sim \text{Unif}\big(S_m(x_1^{(t)},\dots,x_{m-1}^{(t)}; u^{(t)}\big) \propto \mathbb{I}\{K(x_1^{(t)},\dots,x_{m-1}^{(t)}, x_m) > u^{(t)}\}.
\end{aligned}
\label{eq:sgs}
\end{equation}
The scheme in \eqref{eq:sgs} is the auxiliary-variable Gibbs
sampler of \cite{damlen1999gibbs} with a single slice variable shared
across all coordinates. Our contribution is not this decomposition itself
but a fully automated and exact way of carrying out each conditional
update for an arbitrary kernel. The main computational challenge is sampling from the $j$th full conditional given above in eq. (\ref{eq:sgs}). Next, we propose a generic automated method to generate a sample from $j$-th level set given in eq.(\ref{eq:levelset}). For notational convenience, we suppress the conditional values of the rest of the coordinates and denote the generic conditional kernel by $g(z) = K(x_1,\ldots,x_{j-1},z,x_{j+1},\dots,x_m)$. The problem then reduces to generating samples uniformly for the conditional level set $S_u = \{z\in\mathbb{R}: g(z)>u$\}. 

\begin{algorithm}[t]
\caption{ASG sampler for an unnormalized kernel on $\mathbb{R}^m$}
\label{alg:asg}

\begin{algorithmic}[1]
\Require Kernel $K:\mathbb{R}^m\to[0,\infty)$ with $\int_{\mathbb{R}^m}K(x)\,dx<\infty$.
\Require Initial state $x^{(0)}\in\mathbb{R}^m$ with $K(x^{(0)})>0$.
\Require Retained samples $N\in\mathbb{N}$, burn-in $B\ge 0$, thinning $L\ge 1$.
\Require Tolerance $\varepsilon\in(0,1)$, Cauchy scale $s_0>0$.
\Ensure $\{x^{(B+L)},x^{(B+2L)},\dots,x^{(B+NL)}\}$ from a Markov chain with
        stationary distribution $p(x)\propto K(x)$.
\State $T \gets B+NL$
\For{$t=1$ \textbf{to} $T$}
  \State \textbf{Slice height:} draw $u^{(t)}\sim\mathrm{Unif}\bigl(0,K(x^{(t-1)})\bigr)$.
  \For{$k=1$ \textbf{to} $m$}
    \State $g_k(z)\gets K\bigl(x^{(t)}_1,\dots,x^{(t)}_{k-1},\,z,\,x^{(t-1)}_{k+1},\dots,x^{(t-1)}_m\bigr)$
    \State \textbf{Effective support:}
           $[a_k,b_k]\gets\textsc{EffectiveSupport}\bigl(g_k,\ \varepsilon/m,\ s_0\bigr)$
           \Comment{Algorithm~\ref{algo:effective}}
    \If{$x^{(t-1)}_k\notin[a_k,b_k]$}
      \State $x^{(t)}_k\gets x^{(t-1)}_k$ \Comment{retain rule}
    \Else
      \State $x^{(t)}_k\gets\textsc{Slice1D\_FixedU}\bigl(g_k,\ u^{(t)},\ a_k,\ b_k\bigr)$
             \Comment{Proposition~\ref{prop:rejection}}
    \EndIf
  \EndFor
\EndFor
\State \Return $\{x^{(B+L)},x^{(B+2L)},\dots,x^{(B+NL)}\}$
\end{algorithmic}
\end{algorithm}
\subsection{Generic Method for Uniformly Sampling from 1D Slices}
\label{sec: 1dslice}
Given a generic kernel $g(z)$ on $\mathbb{R}$ and a slice height $u>0$, define the level set
$\mathcal{S}_u=\{z\in\mathbb{R}: g(z)>u\}.$
Its Lebesgue measure
$A(u) = \int \mathbb{I}\{g(z)>u\}\,dx$
is the ``width'' of the slice. Uniform sampling from $\mathcal{S}_u$ yields the desired conditional distribution. Our algorithm computes an effective support $[a,b]$ of the conditional kernel, that is, an interval carrying normalized probability mass
$1-\varepsilon$, and then samples uniformly from $S_u\cap[a,b]$. Since
$[a,b]$ is defined through probability mass rather than through the
level set, it need not contain $S_u$; Proposition~\ref{prop:exact} shows that the update is kept exactly invariant. Notice that if $g(z)$ is unbounded, we may have to find the largest $a$ and smallest $b$ such that $\int_a^b g(z) dz\geq (1-\epsilon)\int g(z) dz$, for a chosen very small tolerance $\epsilon>0$. It easily follows that $A(u)$ is a decreasing function for $u>0$ and $\int_0^\infty A(u) du = \int g(z) dz<\infty.$

\subsubsection*{Effective Support Estimation for a Univariate Kernel}
\label{effe_supp}
Let $K:\mathbb{R}\to[0,\infty)$ denote an unnormalized probability kernel. Our objective is to determine an \emph{effective support interval} $[a(\varepsilon), b(\varepsilon)]$ that contains probability mass $1-\varepsilon$ under the normalized density induced by $K(\cdot)$, with equal tail probabilities of size $\varepsilon/2$ on the left and right. Formally, letting 
\[
Z=\int_{-\infty}^{\infty} K(x)\,dx,\qquad f(x)=\frac{K(x)}{Z},
\qquad F(x)=\int_{-\infty}^x f(t)\,dt,
\]
we seek the largest $a=a(\varepsilon)\in\mathbb{R}$ and the smallest $b=b(\varepsilon)>a(\varepsilon)$ such that
\begin{equation}\label{eq:eff}
\int_a^b K(x)\,dx \;\ge\;(1-\varepsilon)\int_{-\infty}^{\infty}K(x)\,dx,
\end{equation}
with the additional equal–tailed requirements
\[
\int_{-\infty}^{a(\varepsilon)} K(x)\,dx \;\le\; \frac{\varepsilon}{2} Z,
\qquad 
\int_{b(\varepsilon)}^{\infty} K(x)\,dx \;\le\; \frac{\varepsilon}{2} Z.
\]
\subsubsection*{Cauchy Transformation and Support Extraction}

Direct evaluation of $F$ and its inverse can be numerically unstable, especially for heavy-tailed or sharply peaked kernels. To stabilize computation, we employ a heavy-tailed change of variables based on a Cauchy distribution with scale parameter $s_0>0$. Let $F_C$, $f_C$, and $Q_C(u)=F_C^{-1}(u)$ denote the distribution function, density, and quantile function, respectively of a Cauchy distribution. The probability integral transform establishes a smooth bijection
$x = Q_C(u), \quad u = F_C(x), \quad u\in(0,1).$
Define the transformed kernel
\[
K^\star(u) \;=\; \frac{K(Q_C(u))}{f_C(Q_C(u))}\,.
\]
A standard change of variables yields
\[
Z = \int_{-\infty}^{\infty} K(x)\,dx
  = \int_0^1 K^\star(u)\,du,
\]
so that $f^\star(u)=K^\star(u)/Z$ defines a valid density on $(0,1)$ and
\[
F(x)
= \int_{-\infty}^x f(t)\,dt
= \int_0^{F_C(x)} f^\star(u)\,du
=:F^\star(F_C(x)).
\]
Hence, the equal–tailed bounds are obtained by solving
$F^\star(u_a)=\frac{\varepsilon}{2},
\quad
F^\star(u_b)=1-\frac{\varepsilon}{2},$
and mapping the solutions back to $\mathbb{R}$ via
$a = Q_C(u_a), 
\quad 
b = Q_C(u_b).$
The resulting interval $[a,b]$ is therefore an equal–tailed credible region of mass $1-\varepsilon$ for the density induced by $K(\cdot)$. The Cauchy transformation is particularly advantageous for kernels with heavy tails or extreme skewness, as its own heavy-tailed nature distributes probability mass more evenly across $(0,1)$ and mitigates numerical underflow. Other monotone transforms with sufficiently heavy tails may be used in principle, though the Cauchy map (if necessary tuning the scale $s_0>0$) has proved especially stable and is widely used in Monte Carlo reparameterization methods \cite{robert1999monte,robert2010introducing}.

%\subsubsection*{Numerical Implementation}
\begin{remark}
In practice, we compute $Z=\int_0^1K^\star(u),du$ and
$F^\star(u)=\int_0^u f^\star(s),ds$ by adaptive quadrature with relative
tolerance $\varepsilon/50$, keeping numerical error small relative to the
target tail mass $\varepsilon/2$. Since $F^\star$ is nondecreasing, the
equations $F^\star(u)=\varepsilon/2$ and $1-F^\star(u)=\varepsilon/2$ are
solved by bracketed root finding \citep{brent2013algorithms}, which is
robust to flat regions of zero density. If the Cauchy-transformed
quadrature is unstable (e.g., for pathological kernels; Fig.~\ref{com_densi}),
we instead identify a finite region where $K(x)$ is numerically
non-negligible, integrate $K$ directly, construct $F$ numerically, and
compute the equal-tailed bounds on the original scale. 
\end{remark}

%In practice, we evaluate the normalizing constant$Z=\int_0^1K^\star(u)\,du$ and the transformed distribution function$F^\star(u)=\int_0^u f^\star(s)\,ds$ by adaptive quadrature on $(0,1)$,with relative tolerance $\varepsilon/50$, so that the quadrature error issmall compared with the tail mass $\varepsilon/2$ being targeted. Because $F^\star$ is nondecreasing, the equations $F^\star(u)=\varepsilon/2$ and $1-F^\star(u)=\varepsilon/2$ are solved by bracketed root finding \citep{brent2013algorithms} rather than by minimizing $|F^\star(u)-\varepsilon/2|$, which is not unimodal when $F^\star$ is flat across a region of zero density. If the Cauchy-transformed quadrature becomes unstable (e.g., for kernels exhibiting pathological behavior on $(0,1)$; see Fig.~\ref{com_densi}), we revert to a fallback strategy: we identify a finite region on $\mathbb{R}$ where $K(x)$ is numerically non-negligible on a grid, integrate $K$ directly over that region to estimate $z$, construct $F$ numerically, and compute the equal–tailed bounds in the original domain.\\

Algorithm~\ref{algo:effective} summarizes the complete effective support estimation procedure. The algorithm returns the bounds $a$ and $b$, the normalized density $f$, and normalizing constant $z$.
\begin{algorithm}[!ht]
\caption{Effective Support Computation for a Kernel Function}
\label{algo:effective}

\begin{algorithmic}[1]
\Function{EffectiveSupport}{$g$, $\epsilon = 0.01$, $\text{s}_0 = 1$}
    \State Define $g(u) \gets \text{Vectorized } \text{ker}(u)$
    \State Define $g^*(u) \gets \frac{g(x)}{f_C(x)}$ where $x = \text{quantile\_Cauchy}(u, \text{s}_0)$, $f_C$ is Cauchy density
    \State Compute normalizing constant $Z \gets \int_0^1 g^*(u) \, du$
    \If{$z$ is finite and $z > 0$}
        \State Define $f^*(u) \gets \frac{g^*(u)}{Z}$
        \State Define $F^*(u) \gets \int_0^u f^*(t) \, dt$
        \State Define $f(x) \gets \frac{g(x)}{Z}$
        \State Define $F(x) \gets F^*(\text{cdf\_Cauchy}(x, \text{s}_0))$
        \State Define $B_{\text{lower}}(u) \gets F^*(u) - 0.5 \cdot \epsilon$
        \State Define $B_{\text{upper}}(u) \gets 1 - F^*(u) - 0.5 \cdot \epsilon$
        \State Find $a^* \gets \text{root of } B_{\text{lower}}(u)$ in $[0, 0.5]$
        \State Find $b^* \gets \text{root of } B_{\text{upper}}(u)$ in $[\max(a^*, 0.5), 1]$
        \State $a \gets \text{quantile\_Cauchy}(a^*, \text{s}_0)$
        \State $b \gets \text{quantile\_Cauchy}(b^*, \text{s}_0)$
    \Else
        \State Define test points $x_{\text{test}} \gets$ linspace from $-100$ to $100$
        \State Find support indices where $g(x) > 10^{-10}$
        \State Let $a_{\text{approx}}$, $b_{\text{approx}}$ be min and max of $x_{\text{test}}$ at support indices
        \State Let $a_{\text{test}} \gets a_{\text{approx}}$, $b_{\text{test}} \gets b_{\text{approx}}$
        \State Compute $Z \gets \int_{a_{\text{test}}}^{b_{\text{test}}} g(x) \, dx$
        \State Define $f(x) \gets \frac{g(x)}{Z}$
        \State Define $F(x) \gets \int_{a_{\text{test}}}^x f(t) \, dt$
        \State Define $B_{\text{lower}}(x) \gets F(x) - 0.5 \cdot \epsilon$
        \State Define $B_{\text{upper}}(x) \gets 1 - F(x) - 0.5 \cdot \epsilon$
        \State Find $a \gets \text{root of } B_{\text{lower}}(x)$ in $[a_{\text{test}}, b_{\text{test}}]$
        \State Find $b \gets \text{root of } B_{\text{upper}}(x)$ in $[a, b_{\text{test}}]$
    \EndIf
    \State Attach attributes to $f$: norm.const $= Z$, lower $= a$, upper $= b$
    \State \Return $\{ \text{lower}: a, \text{upper}: b, f: f, \text{norm.const}: Z \}$
\EndFunction
\end{algorithmic}
\end{algorithm}
%%%%%%%%%%%%%%%%%%%%%%%%%%%%%%%%%%%%%%%%%%%%%%%%%%%%%%%%%%%%
\section{Stationarity and Geometric Ergodicity of the ASG}\label{sec:station}
%%%%%%%%%%%%%%%%%%%%%%%%%%%%%%%%%%%%%%%%%%%%%%%%%%%%%%%%%%%%
In this section we establish the fact that the ASG sampler satisfies the detailed balance equation establishing the stationarity of the Markov Chain generated by the ASG. In addition, under some additional regularity conditions we establish the geometric ergodicity of the ASG sampler.
In particular, we will show that the ASG sampler has the following as its stationary density:
\begin{equation}
H(u, x_1, \cdots, x_m)=\frac{\mathbb{I}\{K(x_1, \cdots, x_m)>u\}}{Z}
\label{eq:stationary} 
\end{equation}
which, in turn ensures the marginal chain in $x$ has stationary distribution $p(x)$.
\begin{proposition}[Exact invariance of the bracketed coordinate update]
\label{prop:exact}
Fix a coordinate $k$, the remaining coordinates $x_{-k}$, and a slice height
$u>0$. Let $g(z)=g_k(z;x_{-k})$ and $S_u =S_u(x_{-k})=\{z\in\mathbb{R}:g(z)>u\}$ with
$0<\lambda(S_u)<\infty$, where $\lambda$ denotes Lebesgue measure. Let
$[a,b]$ be any interval that may depend on $x_{-k}$ and on $g_k$ but not on the current value $x_k$; in particular, the effective-support interval returned by
Algorithm~\ref{algo:effective}. Define the update $x_k\mapsto x_k'$ by
\[
x_k' =
\begin{cases}
x_k, & x_k\notin[a,b],\\[2pt]
y\sim\mathrm{Unif}\bigl(S_u\cap[a,b]\bigr), & x_k\in[a,b],
\end{cases}
\]
where $Y$ is generated by the rejection scheme of
Proposition~\ref{prop:rejection} with $S_u$ replaced by $S_u\cap[a,b]$.
Then $\mathrm{Unif}(S_u)$ is invariant for this update, whether or not
$S_u\subseteq[a,b]$.
\end{proposition}
\begin{proof}
Write $I=[a,b]$ and let $x_k\sim\mathrm{Unif}(S_u)$. If
$\lambda(S_u\cap I)=0$, then $P(x_k\in I)=0$, so $x_k'=x_k$ almost surely
and the claim is immediate. Otherwise, for any measurable $A\subseteq S_u$,
\begin{align*}
P(x_k'\in A)
= P(x_k\in A\cap I^c)
 + P(x_k\in I)\,P\bigl(y\in A\mid x_k\in I\bigr)\\
= \frac{\lambda(A\cap I^c)}{\lambda(S_u)}
 + \frac{\lambda(S_u\cap I)}{\lambda(S_u)}\cdot
   \frac{\lambda(A\cap I)}{\lambda(S_u\cap I)}
 = \frac{\lambda(A)}{\lambda(S_u)}.
\end{align*}
%Because $[a,b]$ does not depend on $x_k$, the two cases partition the state space deterministically given $(x_{-k},u)$, and the result follows.
Because $I$ is fixed conditional on $(x_{-k},u)$ and does not depend on
the current value $x_k$, the sets $I$ and $I^c$ form a deterministic
partition of the $x_k$-state space. Hence the update preserves the
uniform distribution on $S_u$, regardless of whether $S_u$ is contained
in $I$.
\end{proof}

\begin{corollary}\label{cor:exact}
Each coordinate update of Algorithm~\ref{alg:asg} leaves the joint density
$H(u,x)$ in \eqref{eq:stationary} invariant. Consequently the systematic-scan
sweep, and hence the marginal chain $\{x^{(t)}\}$, has stationary
distribution $p(x)\propto K(x)$, regardless of whether
$S_k(u^{(t)})\subseteq[a_k,b_k]$ at any given step.
\end{corollary}
\begin{proof}
Conditional on $(x_{-k},u)$, the $k$-th full conditional of $H$ is
$\mathrm{Unif}\bigl(S_k(x_{-k};u)\bigr)$, which is preserved by
Proposition~\ref{prop:exact}; the $u$-update is an exact Gibbs step.
A composition of $H$-invariant kernels is $H$-invariant, and integrating
out $u$ gives $p(x)$.
\end{proof}

\subsubsection*{Transition kernel of the ASG sampler}
Let $(u,x_1, \cdots, x_m)~\sim H(u, x_1, \cdots, x_m)$ and consider generating $(\tilde{u}, \tilde x_1,\ldots,\tilde x_m)$ using the ASG sampler described in eq. (\ref{eq:sgs}). Notice that
\begin{align*}
\tilde u & \sim \frac{\mathbb{I}\{K(x_1, \cdots, x_m)>\tilde u\}}{K(x_1,\cdots,x_m)} = Unif(0, K(x_1,\cdots, x_m)),\\
\tilde x_1 & \sim \frac{\mathbb{I}(K(\tilde x_1, x_2, \cdots, x_m)>\tilde u)}{A_1(x_2, \cdots, x_m, \tilde u)}~\text{where~} A_1(x_2,\cdots, x_m, \tilde u) = \int \mathbb{I}(K(x_1,x_2,\cdots, x_m)>\tilde u)dx_1
\end{align*}
More generally for $j = 2,\cdots,m-1$, we sample
\begin{align*}
\tilde x_j & \sim \frac{\mathbb{I}(K(\tilde x_1, \cdots \tilde x_j, x_{j+1} \cdots, x_m)>\tilde u)}{A_j(\tilde x_1, \tilde x_{j-1}, x_{j+1} \cdots, x_m, \tilde u)}\\
&\text{where~} A_j(\tilde x_1, \cdots \tilde x_{j-1}, x_{j+1} \cdots, x_m, \tilde u) = \int \mathbb{I}(K(\tilde x_1,\cdots \tilde x_{j-1}, x_j, x_{j+1}\cdots, x_m)>\tilde u)dx_j,
\end{align*}
and finally we sample,
\[
 \tilde x_m \sim \frac{\mathbb{I}(K(\tilde x_1, \cdots, \tilde x_m)>\tilde u)}{A_m(\tilde x_1, \cdots, \tilde x_{m-1}, \tilde u)}~
\text{,where~} A_m(\tilde x_1, \cdots, \tilde x_{m-1}, \tilde u) = \int \mathbb{I}(K(\tilde x_1,\cdots, \tilde x_{m-1}, x_m)>\tilde u)dx_m
\]
Thus the transition kernel of the ASG sampler is given by,
\begin{equation}\label{eq:trans}
\begin{split}
T((u,x_1, \cdots, x_m)\to (\tilde{u},\tilde{x}_1, \cdots \tilde x_m)) =
\frac{\mathbb{I}(K(x_1, \cdots, x_m)> \tilde{u})}{K(x_1, \cdots, x_m)}
\prod_{j=1}^m
\frac{\mathbb{I}\{K(\tilde{x}_1,\dots,\tilde{x}_{j}, x_{j+1}\dots,{x}_m)> \tilde{u}\}}
{A_j(\tilde{x}_1,\dots,\tilde{x}_{j-1}, x_{j+1}\dots,{x}_m, \tilde u)}
\end{split}
\end{equation}
Next we establish the stationarity of the Markov Chain corresponding to the transition kernel given by (\ref{eq:trans}) and show that it's stationary distribution if the joint kernel of $(u, x)$ given by (\ref{eq:stationary}).
\setcounter{theorem}{0}
\begin{theorem}[Stationarity of ASG sampler]
Consider the ASG sampler given in (\ref{eq:sgs}) based on any arbitrary target kernel given in (\ref{eq: K(x)}). Then the corresponding transition kernel given by (\ref{eq:trans}) of produces a Markov Chain with stationary kernel given by (\ref{eq:stationary}).
\label{th:stationary}
\end{theorem}
\begin{proof}
To establish the stationarity of the ASG we need to show the following:
\begin{equation}\label{eq:eq}
\begin{split}
\int T((u,x_1, \cdots, x_m)\to (\tilde{u},\tilde{x}_1, \cdots \tilde x_m))H(u,x_1, \cdots,x_m) du~dx_1\cdots dx_m  = H(\tilde u, \tilde x_1, \cdots, \tilde x_m) \quad \quad
\end{split}
\end{equation}
Notice that the integrand can be expressed as
\[
\frac{\mathbb{I}(K(x_1,x_2,\ldots,x_m)>u)}{K(x_1,\ldots,x_m)}\times 
\prod_{j=1}^{m}
\frac{\mathbb{I}\{K(\tilde{x}_1,\dots,\tilde{x}_{j-1}, x_{j}\dots,{x}_m)\ge \tilde{u}\}}
{A_j(\tilde{x}_1,\dots,\tilde{x}_{j-1}, x_{j+1}\dots,{x}_m, \tilde u)}\times  \frac{\mathbb{I}(K(\tilde x_1, \cdots, \tilde x_m)> \tilde u)}{Z}
\]
Clearly, by using eq.(\ref{eq:lebesgue}), we can see that
\[
\int \frac{\mathbb{I}(K(x_1,x_2,\ldots,x_m)>u)}{K(x_1,\ldots,x_m)} du =1\;\;\;\text{and}\;\;\;
\int \frac{\mathbb{I}(K(x_1,\cdots , x_m)>\tilde{u})}{A_1(x_2,\ldots,x_m, \tilde{u})}dx_1 = 1
\]
\[
\int \frac{\mathbb{I}(K(\tilde x_1,\cdots, \tilde x_{j-1} , x_j,\cdots, x_m)>\tilde u)}{A_j(\tilde x_1, \cdots, \tilde x_{j-1}, x_{j+1}, \cdots, x_m, \tilde u)}dx_j = 1\;\;\text{for}\;j=2,\ldots,m.
\]
And finally by definition, the last term in the integrand is given by
\[
H(\tilde u , \tilde x_1, \cdots \tilde x_m) = \frac{\mathbb{I}(K(\tilde x_1,\cdots,\tilde x_m)>\tilde u)}{Z}
\]
So by using the above equations we can say that equation~(\ref{eq:eq}) is satisfied establishing the stationarity of the ASG sampler.
\end{proof}
\setcounter{theorem}{0}
\begin{remark}
It may be noted that ASG sampler (like any Gibbs sampler) does not necessarily produce a reversible Markov Chain. But by using a random scan version of ASG we can construct a reversible version of ASG. That is after the first step $\tilde u \sim Unif(0, K(x_1,\cdots, x_m))$ we can choose a random permutation of $\{ 1,2,\cdots, m\}$ of the coordinates and call it $\{j_1,j_2,\cdots,j_m\}$ and update $x_{j_1},x_{j_2}, \cdots, x_{j_m}$ as follows.
\[
\tilde x_{j_1} \sim \frac{\mathbb{I}(K(x_1,x_2, \cdots, \tilde x_{j_1}, x_{j_1+1}, \cdots x_m)>\tilde{u})}{A_1(x_1,x_2, \cdots, x_{j_1-1},x_{j_1+1}, \cdots, x_m,\tilde u)}
\]
and then $\tilde x_{j_2},\cdots, \tilde x_{j_m}$ depending on the positions of the coordinates $j_2, \cdots, j_m$.
For other options to convert the systematic scan ASG to random scan ASG sampler we can use techniques similar to those that appears in the Chapter 10 of~\cite{robert1999monte}-Algorithm A.41.
\end{remark}
Under some regularity conditions we establish the uniform geometric ergodicity of the ASG sampler. The primary step involves verifying the {\em Doeblin minorization condition.}. 
Consider a random sample $(U, X_1,\ldots,X_m) \sim H(u, x_1,\ldots,x_m)$ given by (\ref{eq:stationary}). Then the full conditional density of $X_j$ given the $X_1=x_1,\ldots,X_{j-1}=x_{j-1}, X_{j+1}=x_{j+1},\ldots,X_m=x_m, U=u$ is given by 
\begin{equation}
 q_j(x_j \mid x_1,\ldots,x_{j-1},x_{j+1},\ldots,x_m,u) = \frac{\mathbb{I}(K(x_1,\cdots, x_{j-1} , x_j,\cdots, x_m)> u)}{A_j(x_1, \cdots, x_{j-1}, x_{j+1}, \cdots, x_m, u)}
 \label{eq:condl.density}
\end{equation}
Notice that for the SG sampler we will need this density to be positive for any given values of $x_1, \cdots, x_{j-1}, x_{j+1}, \cdots, x_m, u$ in the above conditional density. For the uniform ergodicity, we need a more stronger condition to ensure the positivity of the above conditional density.

\setcounter{theorem}{1}
\begin{theorem}[Uniform ergodicity of ASG sampler]\label{th:ergodicity}
Consider the ASG sampler given in (5) based on a target kernel given in (1)
that satisfies the following additional conditions:
\begin{enumerate}
\item[(A1)] The density is bounded above, i.e., $p(x)\le M$ for some $M>0$
      and all $x\in S=\{x\in\mathbb{R}^m : K(x)>0\}$.
\item[(A2)] For $j=1,\dots,m$, the conditional density in (10) satisfies
      $q_j(x_j\mid x_1,\dots,x_{j-1},x_{j+1},\dots,x_m,u)\ge\delta_j$ for
      some $\delta_j>0$, uniformly for any $u>0$ and any $x\in S$ with
      $K(x)>u$.
\item[(A3)] $S$ is a product of intervals, and $p(x)\ge\delta_0$ for some
      $\delta_0>0$ and all $x\in S$.
\end{enumerate}
Then the transition kernel
$T(x\to\tilde x)=\int_0^\infty T((u,x)\to(\tilde u,\tilde x))\,d\tilde u$
satisfies the Doeblin condition $T(x\to\tilde x)\ge\epsilon\,p(\tilde x)$
for all $x,\tilde x\in S$, where
$\epsilon=\delta_0\prod_{j=1}^m\delta_j/M^2$. As a consequence, the ASG is
uniformly ergodic [e.g., by Theorem 6.59 given in \cite{robert1999monte}].
\end{theorem}

\begin{proof}
From (8), since $\tilde u\sim\mathrm{Unif}(0,K(x))$, the transition kernel
of the ASG sampler can be written as
\[
T(x\to\tilde x)=\int_0^{K(x)}\frac{1}{K(x)}\prod_{j=1}^m
q_j(\tilde x_j\mid \tilde x_1,\dots,\tilde x_{j-1},x_{j+1},\dots,x_m,\tilde u)
\,d\tilde u .
\]
Fix $x,\tilde x\in S$. Since $S$ is a product of intervals by (A3), each
intermediate point $(\tilde x_1,\dots,\tilde x_j,x_{j+1},\dots,x_m)$ also
lies in $S$, and so $K\ge Z\delta_0$ there. Hence, for every
$\tilde u<Z\delta_0$, each intermediate point lies in the slice
$\{K>\tilde u\}$, and by (A2) each $q_j\ge\delta_j$. Keeping only
$\tilde u<\min\{K(x),Z\delta_0\}$ in the integral,
\[
T(x\to\tilde x)\ \ge\ \prod_{j=1}^m\delta_j\,
\frac{\min\{K(x),Z\delta_0\}}{K(x)}
\ \ge\ \prod_{j=1}^m\delta_j\,\frac{\delta_0}{M},
\]
where the last step uses $K(x)=Z\,p(x)\le ZM$ by (A1) and
$\delta_0\le M$. Finally, $p(\tilde x)\le M$ by (A1), so
$1\ge p(\tilde x)/M$, and therefore
$T(x\to\tilde x)\ge\bigl(\delta_0\prod_{j}\delta_j/M^2\bigr)\,p(\tilde x)
=\epsilon\,p(\tilde x)$.
\end{proof}
\begin{remark}
Condition (A2) forces every conditional level set to have length at most
$1/\delta_j$, so together with (A3) the theorem applies to kernels that
are bounded above and away from zero on a bounded rectangle, e.g.\ the
Rosenbrock and Ackley kernels truncated to a rectangle. It does not cover
kernels with unbounded support (the Euclidean-norm, LASSO and Bridge
kernels) or supports with gaps (the Beta mixture); for these, validity
rests on the invariance result of Proposition 3.1. Weakening
(A2) and (A3) remains a part of future explorations.
\end{remark}
This stands in sharp contrast to gradient--based methods such as HMC or to random--walk and adaptive RW-MH algorithms, for which geometric ergodicity may fail or can be difficult to verify, and for which convergence cannot be assessed from any low--dimensional functional of the chain alone. Consequently, slice--based algorithms offer a uniquely transparent and theoretically grounded mechanism for both \emph{monitoring} and \emph{guaranteeing} convergence, representing a major methodological advantage in high--dimensional or irregular settings where traditional MCMC diagnostics can be opaque or unreliable.

\subsection{Convergence Diagnostics of Slice Samplers}
An additional and often underappreciated advantage of slice--based MCMC methods is the simplicity with which convergence diagnostics may be conducted. Because each iteration of a slice sampler operates by sampling uniformly from the region $\{x : K(x) \ge u\}$ for an auxiliary slice height $u$, \emph{stationarity of the $m$--dimensional Markov chain is equivalent to stationarity of the log--kernel values} $\log K(x^{(t)})$ for $t=1,2,\ldots$ observed along the chain. Thus, instead of monitoring potentially $m$-dimensional trajectories of $\{x^{(t)}; t=1,2,\ldots\}$, one may diagnose convergence by examining the univariate series $\{\log K(x^{(t)}); t=1,2\ldots\}$, whose marginal distribution at stationarity is fully determined by the underlying target density $f(x)\propto K(x)$. This equivalence arises because the slice sampler constructs a uniformly ergodic Markov chain on the joint $\{(x^{)t)}, u^{(t)}); t=1,2,\ldots\}$ space whose marginal $\{x^{(t)};\;t=1,2,\ldots\}$ has the correct invariant distribution, and the uniformity on slices ensures that the evolution of the ordinate variable faithfully reflects the chain’s approach to equilibrium. 
Theoretical results by \cite{roberts2001optimal,robertsrosenthal1999} show that, under mild regularity conditions considerably weaker than those typically required for MH or Hamiltonian Monte Carlo slice samplers are \emph{geometrically ergodic}, thereby providing a rigorous foundation for such diagnostics (see Theorem~6 of \cite{miratierney2002}. Theorem~\ref{th:ergodicity} shows that the proposed ASG is uniformly ergodic under conditions (A1)--(A3).

\subsubsection*{Multivariate Effective Sample Size (multiESS)}
\label{benchmark}
Univariate ESS quantifies mixing efficiency for a single scalar
summary $h(X)$, as defined in Effective sample size. For a
multivariate target $\theta\in\mathbb{R}^p$, we additionally report the
multivariate effective sample size (multiESS) of \cite{vats2019multivariate},
implemented in the \texttt{mcmcse} package \citep{flegal2017mcmcse}, which
accounts for the full covariance structure and cross-coordinate
autocorrelation of the chain rather than $p$ separate univariate ESS
values.\\
Let $\{\theta^{(t)}\}_{t=1}^n$ denote the post-burn-in MCMC output, with
sample covariance matrix
\[
  \Lambda \;=\; \frac{1}{n-1}\sum_{t=1}^{n}
  \bigl(\theta^{(t)}-\bar\theta\bigr)\bigl(\theta^{(t)}-\bar\theta\bigr)^\top,
  \qquad
  \bar\theta = \frac{1}{n}\sum_{t=1}^n \theta^{(t)},
\]
and let $\Sigma$ denote a consistent estimate of the multivariate Monte
Carlo asymptotic covariance matrix (e.g.\ via multivariate batch means
or a spectral variance estimator), the multivariate analogue of
$\sigma_h^2\tau$. The multivariate effective
sample size is then defined as
$  \widehat{\mathrm{mESS}} \;=\; n\left(\frac{|\Lambda|}{|\Sigma|}\right)^{1/p},$
where $p=\dim(\theta)$ and $|\cdot|$ denotes the matrix determinant.
The ratio $|\Lambda|/|\Sigma|$ compares the volume of the sample
covariance to the autocorrelation-inflated covariance; the $p$-th root
rescales this volume ratio to a per-sample rate, and multiplying by
$n$ yields an effective count of independent draws. When $p=1$,
$\widehat{\mathrm{mESS}}$ reduces to the univariate
$\mathrm{ESS}=n\sigma^2/\hat\sigma^2$.

\subsection{Other Samplers Compared with ASG sampler}
For all five baseline samplers, we relied on existing, well-tested R
implementations rather than writing our own: Random Walk
Metropolis--Hastings came from the \texttt{metrop} function in the
\texttt{mcmc} package~\cite{geyer1992practical,Geyer_2011}, Adaptive
Gibbs sampling from \texttt{adaptMetropGibbs} in
\texttt{spBayes}~\cite{Finley_2015, Finley_2020}, Hamiltonian Monte
Carlo from \texttt{hmc} in \texttt{rhmc}~\cite{betancourt2017conceptual},
and both Quantile and Elliptical Slice Sampling~\cite{murray2010elliptical} from the
\texttt{qslice} package's \texttt{slice\_quantile} and
\texttt{slice\_elliptical} functions, respectively~\cite{neal2003slice,
Heiner_2024, heiner2024quantile}. This let us focus our effort on a
fair, comparison rather than reimplementing
established methods.
\paragraph{Scope of replication.} Given the substantially higher sampling
cost of the Euclidean-norm ($m=10$), LASSO, and Bridge Regression
($m=21$) experiments relative to the univariate, Rosenbrock, and
Ackley kernels, replication (Section~\ref{sec:rosenbrock}) was
prioritized for the latter three as representative benchmarks.
Results for the higher-dimensional kernels have been limited to a single run.

\section{Numerical Illustrations Based on Non-standard Probability Kernels}\label{sec:examples}
The proposed algorithm are based on the following setting: samples = 1000, burn = 250, thin = 1. The primary computational modules of the algorithm are summarized below, first the \texttt{effective support} function  estimates the effective support interval [a,b] of the target kernel. Then the  \texttt{Sliced fixed u}  function implements a rejection-based slice sampling procedure embedded within the Gibbs sampling framework.
After that the  \texttt{ASG} function performs coordinate-wise updates across all dimensions, generating posterior draws whose diagnostics are visualized using \texttt{ggplot2}.
Across all experimental settings described below, the diagnostic visualizations consistently demonstrate rapid convergence and stable chain behavior. The proposed sampler exhibits substantially superior mixing properties compared to other counterpart. Moreover, the gain in Multi-ESS remains consistent across both bounded and unbounded target distributions, underscoring the robustness and computational efficiency of the method. 
\subsection{Univariate Kernel}\label{univariate}
The univariate target density considered in this study is constructed as a mixture of three Beta distributions. Formally, the kernel is defined as
\begin{equation}\label{eq:k1}
\begin{split}
K(x) =  \frac{0.3}{2} \text{Beta}\left(\frac{x + 5}{2}; 0.5, 1\right) + \frac{0.4}{2} \text{Beta}\left(\frac{x + 1}{2}; 2, 2\right) +   \frac{0.3}{2} \text{Beta}\left(\frac{x - 5}{2}; 1, 0.5\right)
\end{split}
\end{equation}
where each component Beta density is appropriately shifted and scaled to ensure non-overlapping yet moderately interacting modes. This mixture configuration facilitates the evaluation of the sampler’s ability to efficiently explore multimodal probability landscapes.
\subsubsection{Early-Stage Exploration Behavior of the Effective support of the proposed ASG Sampler}
\begin{figure}[!ht]
    \centering
    \includegraphics[width=0.8\linewidth]{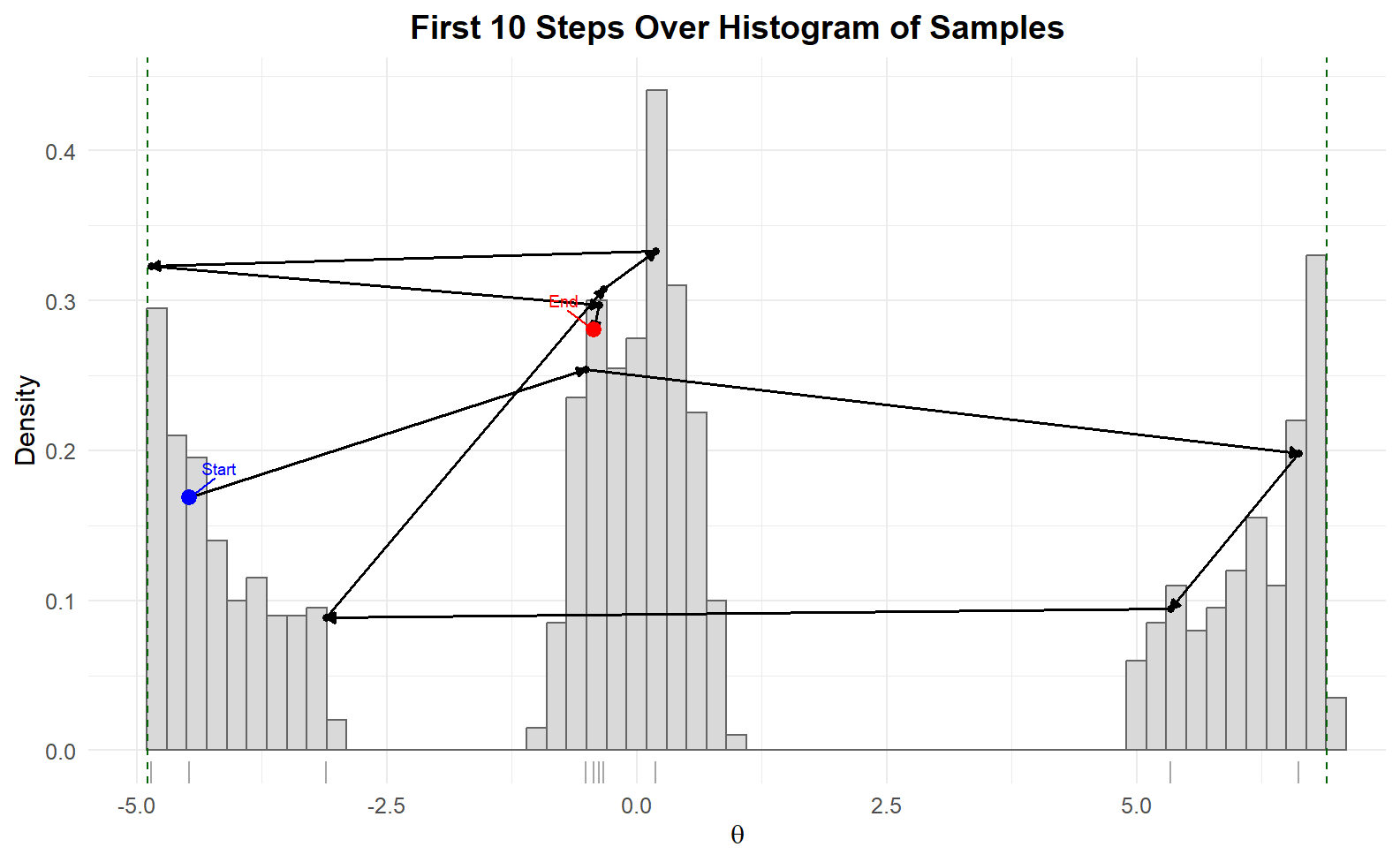}
    \caption{First 10 steps of the ASG algorithm based on the effective support over the beta kernel~\ref{eq:k1}}
    \label{fig:first_20}
\end{figure}
Figure~\ref{fig:first_20} illustrates the behavior of the effective support of our proposed ASG sampling algorithm during its initial ten iterations when applied to the mixture of Beta kernels defined in Equation~\ref{eq:k1}. The gray histogram in the background represents the empirical distribution of all post--burn-in samples, which effectively approximates the target normalized density. The histogram is superimposed on a black connected trajectory depicting the first ten accepted samples of the proposed sampler. The initial sample (black dot) denotes the starting point of the chain, while the red dot marks the tenth accepted sample. Notably, the early samples exhibit pronounced jumps across the three major modes of the mixture, approximately spanning the range $[-7, 7]$, before the chain begins to stabilize within the dominant high-density regions. The dashed green vertical lines indicate the estimated effective-support bounds inferred from the kernel. This figure demonstrates the sampler's capacity for rapid global exploration in the initial phase, efficiently identifying and revisiting the principal probability mass regions of a multimodal target distribution. The visualization thereby highlights both the flexibility and stability of the proposed ASG approach during the burn-in period.

\subsubsection{Numerical Evaluation Metrics of the proposed ASG sampler}
The Figure~\ref{com_densi}(a) illustrates the density sampled overlaid on the true kernel, demonstrating the sampler's ability to capture the target distribution. Table~\ref{tab:comparison_uni} compares the performance of the Proposed ASG sampler and other sampling methods.
\noindent Table~\ref{tab:comparison_uni} presents a comparative analysis between the proposed ASG sampler and the other sampler across four key performance metrics: computational time (in seconds), effective sample size (ESS)n. The ASG sampler demonstrates a markedly higher ESS, reflecting substantially improved chain mixing and a more efficient exploration of the target distribution. Despite this improvement in sampling efficiency, the computational time remains comparable between these methods, indicating that the adaptive support estimation in ASG does not impose a significant computational overhead. The diagnostic plots (in the Appendix)—specifically the trace plot (Figure~\ref{trace_uni}), and autocorrelation function (ACF) plot—corroborate these findings. These visual diagnostics exhibit rapid convergence, low serial autocorrelation, and stable running averages for the ASG sampler, collectively providing strong empirical evidence of its superiority over traditional MCMC methods in terms of both mixing efficiency and reliability of inference.
\begin{table}[!ht]
    \centering
    
    \begin{tabular}{l|r}
    \textbf{Sampler} & \textbf{Multi ESS (mean $\pm$ sd, $R=30$)}  \\
    \hline
    \hline
      \textbf{Proposed ASG}   &  1016 $\pm$  143   \\
      \textbf{RW-MH}   & 3 $\pm$ 1 \\
      \textbf{HMC} &  416 $\pm$ 38 \\
      \textbf{Adaptive Gibbs} &  237  $\pm$ 33 \\
      \textbf{Qslice} &  213 $\pm$ 41 \\
          \textbf{Elliptical slice} &   273 $\pm$ 77   \\
    \end{tabular}
\caption{Comparison of ASG and other algorithms for Univariate Kernel based on $1000$ MCMC samples post burn-in, replicated over $R=30$ independent runs (mean $\pm$ sd of Multi ESS).}
    \label{tab:comparison_uni}
\end{table}
\subsection{Bivariate Kernels}
The goal of implementing the proposed ASG sampler on highly complex bivariate kernels with sharp curvatures is to rigorously assess its ability to adapt to extreme local geometry while maintaining global exploration and statistical efficiency. Such targets present well-known difficulties for conventional MCMC methods due to strong anisotropy, narrow ridges, abrupt changes in curvature, and potentially disconnected high-density regions, which often lead to slow mixing, high autocorrelation, or failure of gradient-based methods. By applying ASG to this challenging class of distributions, we aim to demonstrate that its automated support calibration and slice-driven Gibbs updates can simultaneously achieve robust mode traversal, stable local adaptation, and reliable convergence without problem-specific tuning. The empirical performance of ASG will be systematically compared against widely used baselines—including Random Walk Metropolis–Hastings, Hamiltonian Monte Carlo, Adaptive Gibbs, and Qslice—using effective sample size per second as the primary criterion, alongside additional diagnostics such as autocorrelation decay, acceptance behavior, computational cost, and sensitivity to hyperparameters.
\subsubsection{Rosenbrock kernel (Banana Shaped Density)}\label{sec:rosenbrock}
The bivariate Rosenbrock distribution, often used as a benchmark for evaluating sampling algorithms, is characterized by its highly nonlinear and curved “banana-shaped” contours that pose a significant challenge for standard Markov chain Monte Carlo (MCMC) methods. The distribution is defined through its unnormalized density function as
\begin{equation}\label{eq:rosenbrock kernel}
K(x_1, x_2) = \exp\left( -\frac{x_1^2}{10} - \frac{x_2^2}{10} - 2 (x_2 - x_1^2)^2 \right)
\end{equation}
where $x_1,~x_2 \in \mathbb{R}$. The strong nonlinear dependency between the two parameters creates narrow curved ridges in the probability surface, making it an ideal test case to assess mixing efficiency, convergence behavior, and robustness of advanced sampling strategies. Table~\ref{tab:rosenbrock_replicated} presents the performance metrics.
\begin{figure}[!ht]
    \centering
         \begin{subfigure}[b]{0.45\linewidth}
    \includegraphics[width=\linewidth, height = 3cm]{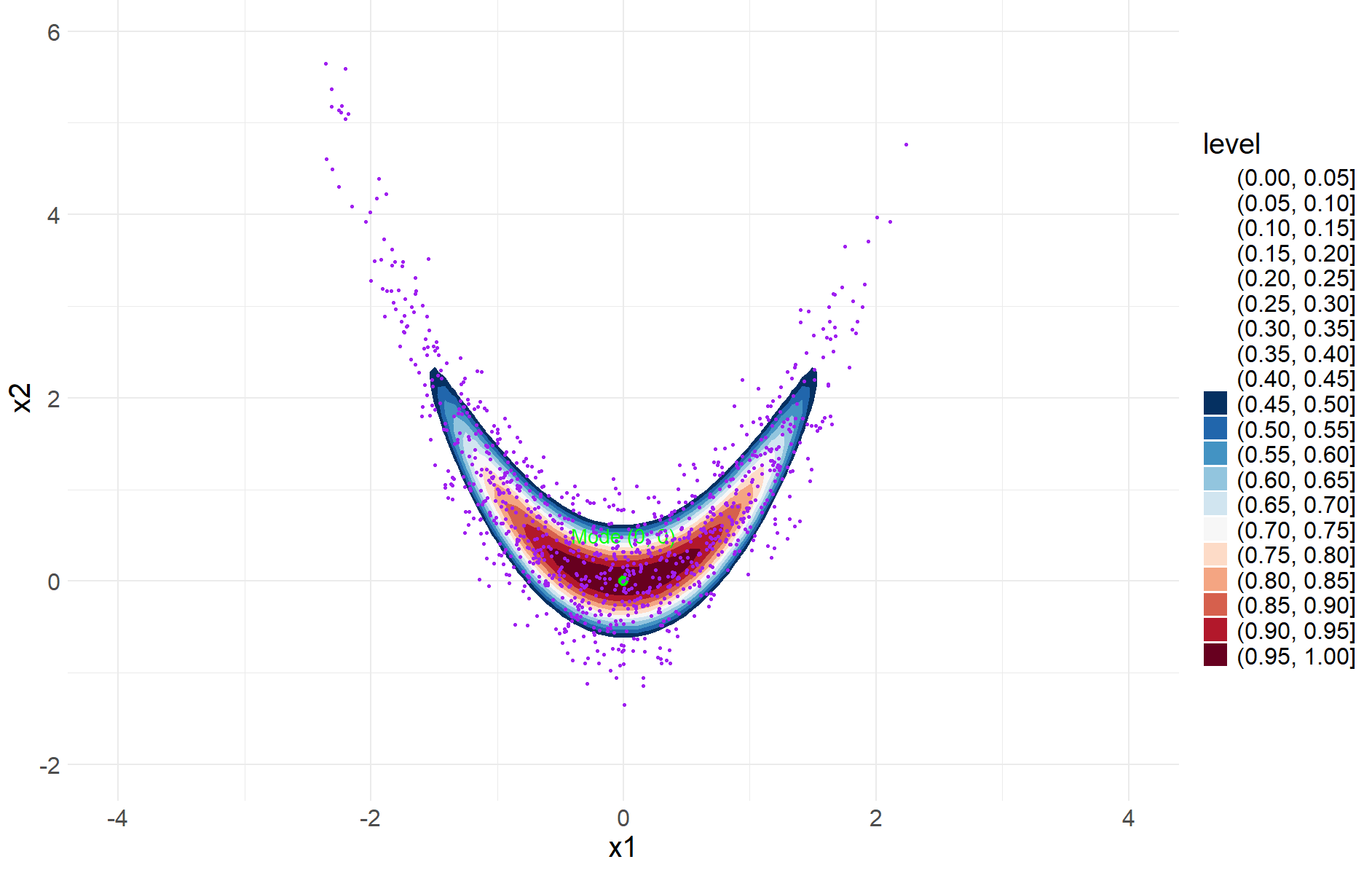}
\end{subfigure}
     \begin{subfigure}[b]{0.45\linewidth}
    \centering
    \includegraphics[width=\linewidth]{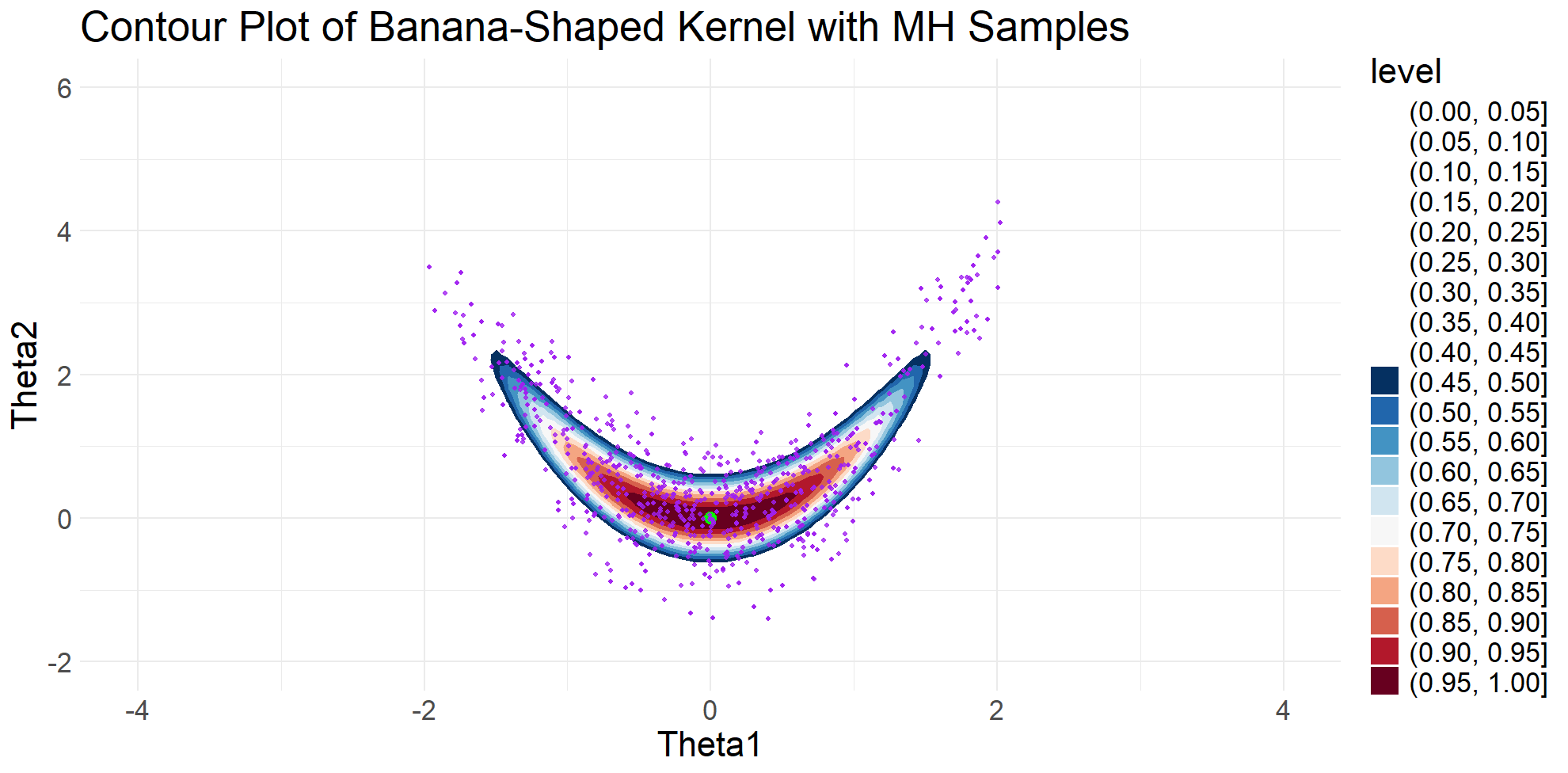}
    \end{subfigure}
    \caption{Samples overlayed on contour Plot of Rosenbrock Kernel, (a) generated from proposed ASG and (b) From RW-MH.}
    \label{compare_ackley}
\end{figure}
The proposed ASG sampler exhibits competitive sampling performance relative to most baseline algorithms, though Elliptical Slice Sampling attains a higher multivariate effective sample size on this smoothly curved, unimodal kernel. This is consistent with the pattern noted in Appendix, ASG's coordinate-wise update structure is most advantageous for multimodal or disconnected supports, whereas methods proposing joint moves along the target's geometry such as elliptical slice sampling's elliptical arcs can traverse smooth curvature more efficiently in a single step. Despite this improvement in sample quality, the ASG sampler achieves these results with slightly lower computational cost, highlighting its computational efficiency. Visual diagnostics further corroborate these findings. The trace plots (Figure~\ref{trace_all_rose}), and ACF plots consistently demonstrate that the proposed ASG sampler converges more rapidly and yields lower autocorrelation compared to the RW-MH sampler. Collectively, these results confirm the superior accuracy, efficiency, and robustness of the ASG sampler in exploring complex, curved probability landscapes.\\
Table~\ref{tab:rosenbrock_replicated} reports Multi ESS averaged over
$R=30$ independent replications of each sampler on the Rosenbrock
kernel. Elliptical Slice Sampling attains the highest average Multi
ESS (414), followed by the proposed ASG sampler (352), Qslice
(160), HMC (45), and Adaptive Gibbs (41); RW-MH's (23) is the lowest. Relative to HMC and Adaptive Gibbs, ASG's
advantage is substantial and the two samplers' sd-scaled ranges do not
overlap, indicating a robust improvement in mixing efficiency on this
kernel despite its curved, unimodal geometry. ASG also outperforms
Qslice on average, though the two methods' variability is large enough that this comparison would
benefit from a formal test before being stated as conclusive.
\begin{table}[!ht]
\centering

\begin{tabular}{lc}
\toprule
\textbf{Sampler} & \textbf{Multi ESS (mean $\pm$ sd, $R=30$)} \\
\midrule
Proposed ASG      & 352 $\pm$ 77 \\
RW-MH & 23 $\pm$ 6\\ 
HMC               & 45 $\pm$ 22 \\
Adaptive Gibbs    & 41 $\pm$ 34 \\
Qslice            & 160 $\pm$ 69 \\
Elliptical slice  & 414 $\pm$ 76 \\
\bottomrule
\end{tabular}
\caption{Comparison of ASG and other algorithms for the Rosenbrock kernel, based on 1000 post-burn-in MCMC samples, replicated over $R=30$ independent runs (mean $\pm$ sd of Multi ESS).}
\label{tab:rosenbrock_replicated}
\end{table}
Elliptical Slice Sampling attains a higher average Multi ESS than ASG on this kernel (414 vs.\ 352);
however, as with the Qslice comparison above, the two samplers' sd-scaled ranges overlap, so this
difference should be interpreted as descriptive rather than as evidence of a robust effect without further
replication or a formal statistical comparison.\\
{In this experiment, the two-dimensional Rosenbrock (banana-shaped) kernel was analyzed to illustrate how the proposed effective-support algorithm adapts its bounds when one variable is fixed. The joint kernel exhibits a curved, non-Gaussian structure centered near $(x_1, x_2) = (0, 0)$, corresponding to the global mode. To investigate conditional behavior, both variables were fixed at selected slice values $c \in \{-2, 0, 1\}$, representing the global mode and two random locations. For the conditional analysis of $X_2 \mid X_1 = c$. The estimated bounds are summarized in Table~\ref{tab:x2_given_x1}, showing that as $x_1$ increases from $-2$ to $1$, the support of $x_2$ contracts and shifts upward (from $[-0.0002,\, 5.0667]$ to $[-0.3045,\, 2.2093]$). This indicates a clear curvature in the density, consistent with the banana-shaped structure. Similarly, for $X_1 \mid X_2 = c$, the horizontal red segments depict the effective bounds of $x_1$ for each fixed $x_2$. Table~\ref{tab:x1_given_x2} reports these ranges, showing near-symmetric intervals about zero, e.g., $[-1.0415,\, 1.0415]$ for $x_2 = 0$. This symmetry verifies that the $x_1$ distribution remains centered and less skewed compared to $x_2$.
}

\begin{table}[!ht]
\centering
\begin{minipage}{0.4\textwidth}
\centering

\begin{tabular}{l||ccc}
\hline
Label & $x_1$ & Lower Bound($y_0$) & Upper Bound($y_1$) \\
\hline
lower   &  0  & -1.256 &  1.256 \\
lower1  & -2  & -0.002 &  5.066 \\
lower2  &  1  & -0.304 &  2.209 \\
\hline
\end{tabular}
\caption{ Vertical bounds: $X_2 \mid X_1 = c$}
\label{tab:x2_given_x1}
\end{minipage}
\hfill
\begin{minipage}{0.43\textwidth}
\centering

\begin{tabular}{lccc}
\hline
$x_2$ & Lower Bound($x_0$) & Upper Bound($x_1$) \\
\hline
  0  & -1.041 &  1.041 \\
 -2  & -0.602 &  0.602 \\
  1  & -1.426 &  1.426 \\
\hline
\end{tabular}
\caption{ Horizontal bounds: $X_1 \mid X_2 = c$}
\label{tab:x1_given_x2}
\end{minipage}
\end{table}

\subsubsection{Ackley Function}
To evaluate the performance of the proposed sampling algorithm, we employ the two-dimensional Ackley function, a well-known non-convex benchmark characterized by a large number of local minima and a single global minimum at the origin. This function serves as an effective testbed for assessing the convergence behavior and robustness of sampling schemes in complex, multimodal landscapes.
\[
f(x) = -a \exp\left(-b \sqrt{\frac{1}{m} \sum_{i=1}^m x_i^2}\right) - \exp\left(\frac{1}{m} \sum_{i=1}^m \cos(c x_i)\right) + a + e,
\]
where \( a = 20 \), \( b = 0.2 \), and \( c = 2\pi \). 
\begin{table}[!ht]
    \centering
    
    \begin{tabular}{l|r}
    \textbf{Sampler} &  \textbf{Multi ESS (mean $\pm$ sd, $R=30$)}  \\
    \hline
    \hline
      \textbf{Proposed ASG }   &  994 $\pm$ 19   \\
      \hline
      \textbf{RW-MH}   &  143 $\pm$ 15  \\
      \hline
                \textbf{HMC } &  96 $\pm$ 59 \\
        \hline
          \textbf{Adaptive Gibbs } & 95 $\pm$ 74 \\
   %     \textbf{Adabtiive Gibbs Sampler (Dim 3):} & & &\\
        \hline
      \textbf{Qslice sampler } &  969 $\pm$ 77 \\
                \hline
      \textbf{Elliptical slice } &  826 $\pm$  144  \\
  %  \textbf{Qslice sampler (Dim 3):} & & &\\684.4685 757.7440
    \end{tabular}
\caption{Comparison of ASG and other algorithms for Ackley Function based on 1000 MCMC samples post burn-in, replicated over $R=30$ independent runs (mean $\pm$ sd of Multi ESS).}
    \label{tab:comparison_ackley}
\end{table}
\begin{figure}[!ht]
\centering
\begin{subfigure}[b]{0.45\linewidth}
    \centering
    \includegraphics[width=1\linewidth]{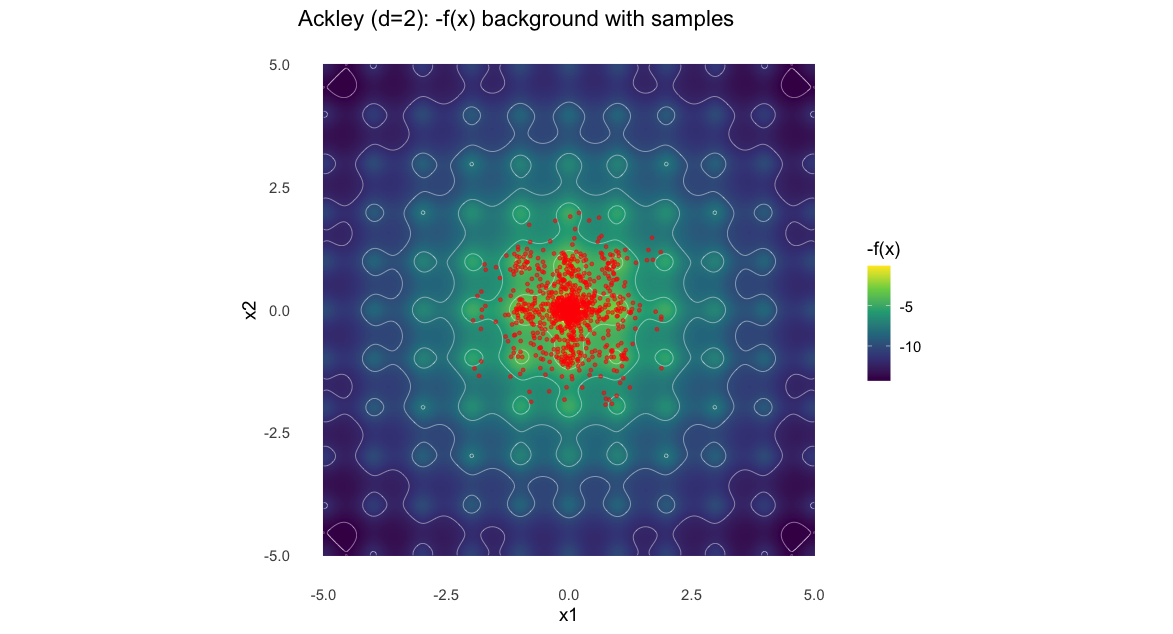}
\end{subfigure}
\begin{subfigure}[b]{0.45\linewidth}
    \centering
    \includegraphics[width=0.9\linewidth]{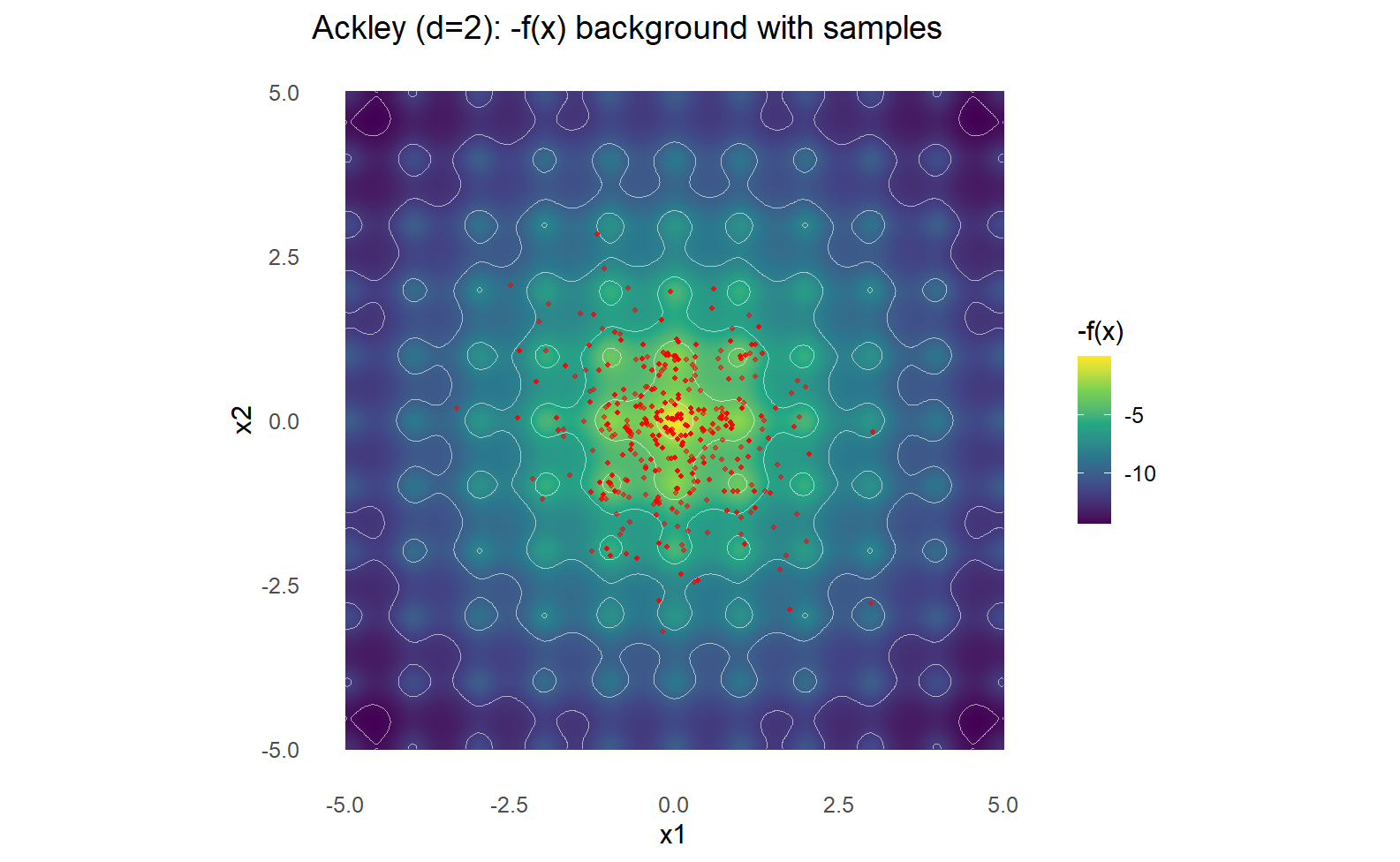}
\end{subfigure}
\caption{Contour Plot of Ackley Function where for both figures the red dots are the samples generated by ASG((a)) and RW-MH((b)) overlayed on the contour plot.}
\label{ach_all}
\end{figure}
\noindent For the challenging multimodal Ackley kernel, the proposed ASG sampler once again exhibits substantially superior sampling performance compared to the other MCMC algorithm. As reported in Table~\ref{tab:comparison_ackley}, the ASG sampler achieves markedly higher Multi ESS values of 994 $\pm$ 19, whereas the other samplers give lower ESS. Despite this significant improvement in sampling efficiency, the computational cost of the ASG sampler
remains modest. This result underscores the robustness and scalability of the proposed method when applied to complex, non-Gaussian, and multimodal target densities. In addition, diagnostic visualizations support these quantitative findings. The trace plots (Figure~\ref{trace_all_ack}), and ACF plots all reveal markedly improved mixing behavior and faster convergence of the ASG sampler relative to the other sampler. Collectively, these results confirm that the ASG approach maintains both efficiency and stability, even in high-dimensional and rugged posterior landscapes.
\subsection{Multivariate Kernels in High Dimension}
Consider the $m$-dimensional density (from \cite{robert1999monte} Example 8.9, p-332)
$K(x) \propto \exp\{-\|x\|\}, \quad x \in \mathbb{R}^m,$
which depends only on the Euclidean norm $\|x\|$.  
Let $z = \|x\|$ denote the radial component.  
Then, the induced marginal density of $z$ is given by $\eta_m(z) \propto z^{m-1} e^{-z}, \quad z > 0,$
that is, a Gamma distribution with shape parameter $m$ and rate parameter $1$.  
By defining $u = z^m$, the corresponding transformed density becomes $\pi_m(u) \propto e^{-u^{1/m}}, \quad u > 0.$
The results for 10-dimensional kernel based on the Multi ESS where the Multi ESS Score is 1000 based on the final 1000 samples following a 200 Burn-in.\\
It has been noted that as the dimension $m$ increases, the standard slice sampler associated with $\pi_m$ experiences a severe degradation in performance.  For smaller dimensions ($m = 1, 5$), the Markov chain mixes well, with the ACF decaying rapidly.   However, for $m = 10$, mixing becomes noticeably slower, and for $m = 50$ the convergence almost ceases, indicating poor exploration of the state space~\cite{robert1999monte}. We compare the performance of the standard slice sampler with the proposed ASG sampler.
To evaluate the representativeness of the MCMC samples drawn from the 10-dimensional kernel using our proposed ASG, the distribution of the sampled radii
$r = \|x\| = \sqrt{x_1^2 + x_2^2 + \cdots + x_{10}^2}$
was compared with the theoretical Gamma$(10,1)$ distribution.
The histogram of the simulated $\|x\|$ values closely matched the red overlaid Gamma$(10,1)$ curve, confirming that the generated samples follow the correct theoretical distribution. Also the ACF of $\|x\|$ decayed rapidly, indicating that the sampler achieved adequate mixing for $m=10$ (see Figure~\ref{radii_comp} in the Appendix).
The results confirm that, while the slice sampler performs well for moderate dimensions, its efficiency declines sharply as $m$ grows.

\begin{figure}[!ht]

    \centering
     \begin{subfigure}[b]{0.49\linewidth}
    \centering
    \includegraphics[width=\linewidth]{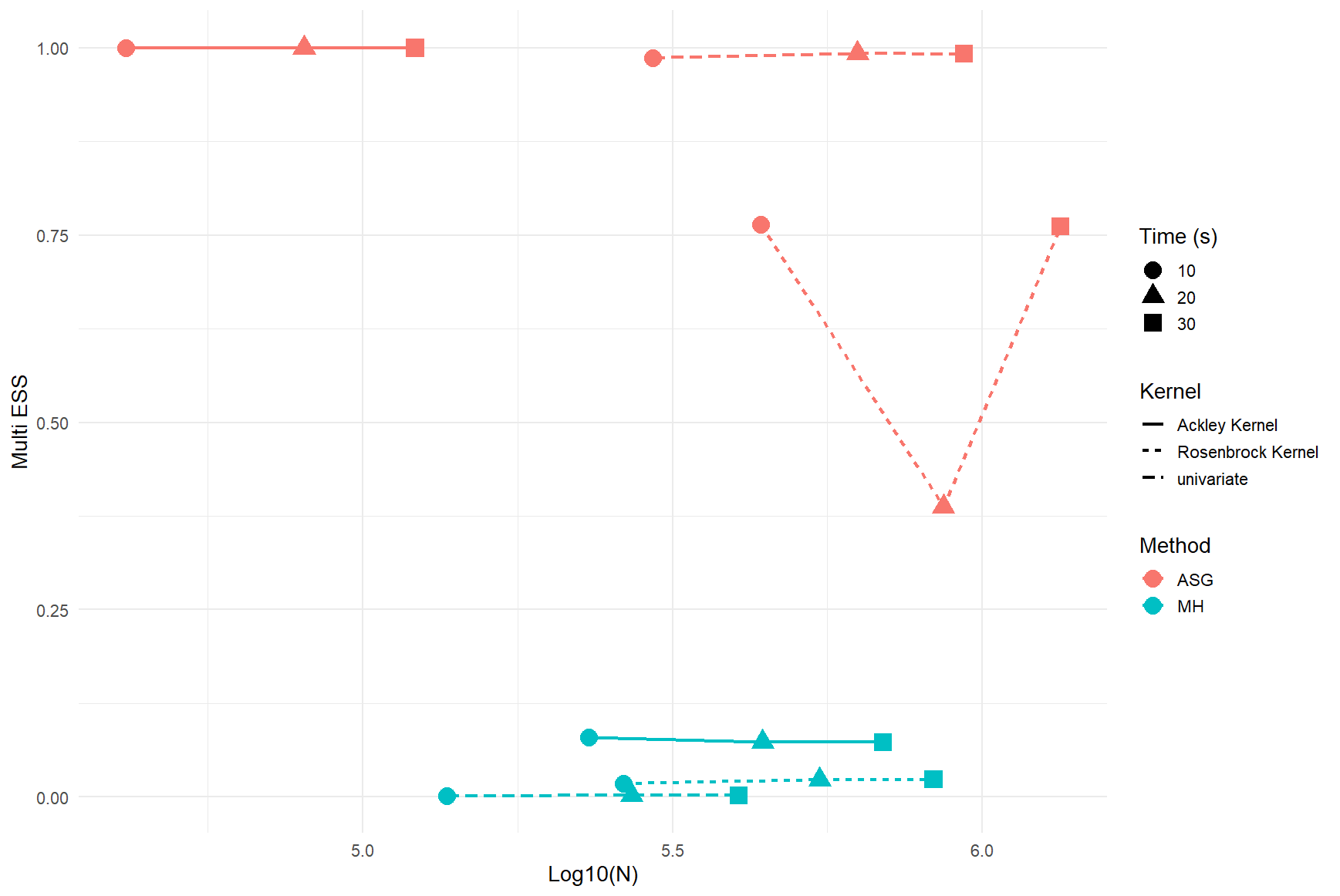}
\end{subfigure}
\hfill
     \begin{subfigure}[b]{0.49\linewidth}
    \centering
    \includegraphics[width=\linewidth]{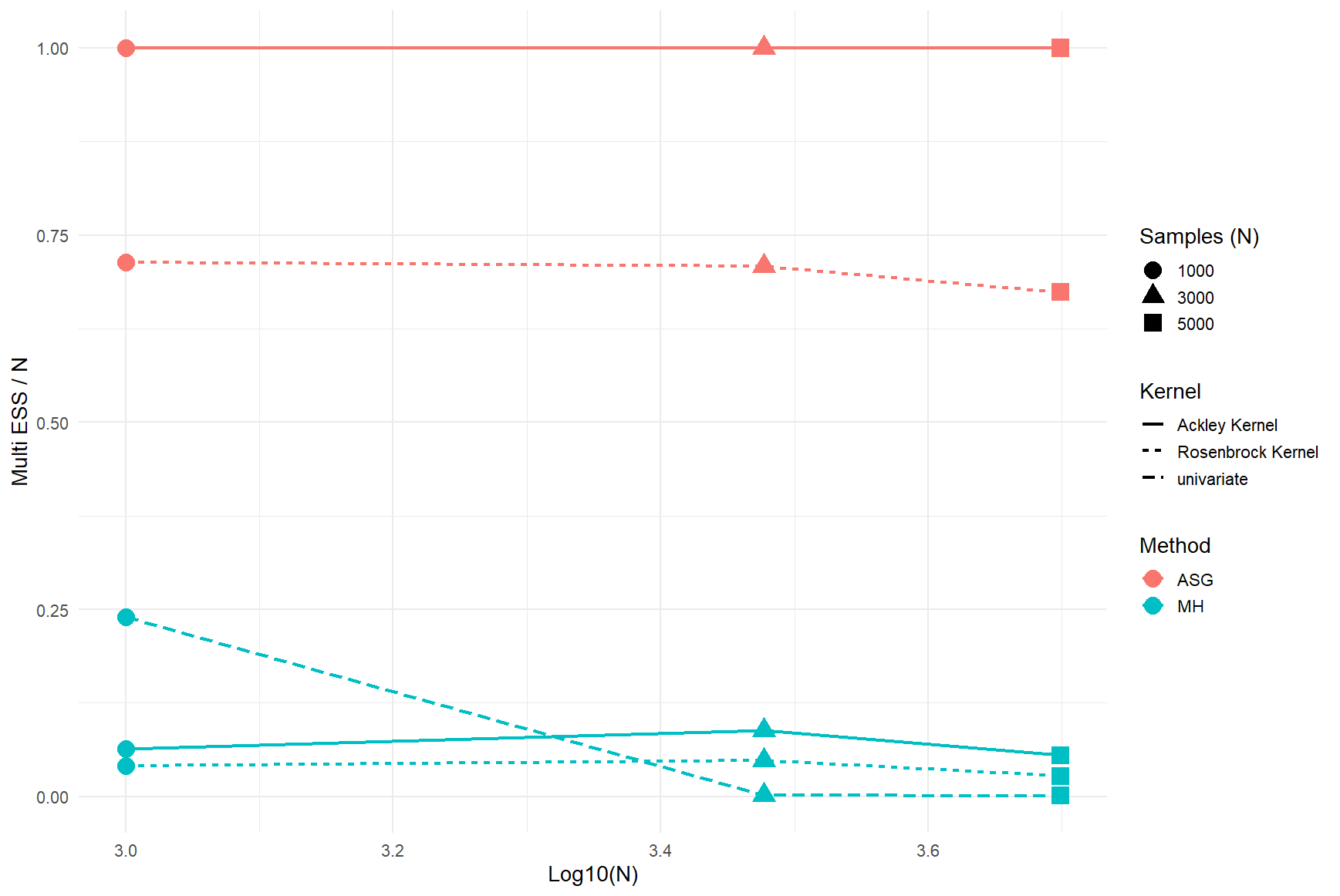}
    \end{subfigure}
\hfill
\begin{subfigure}[b]{0.49\linewidth}
    \centering
    \includegraphics[width=\linewidth]{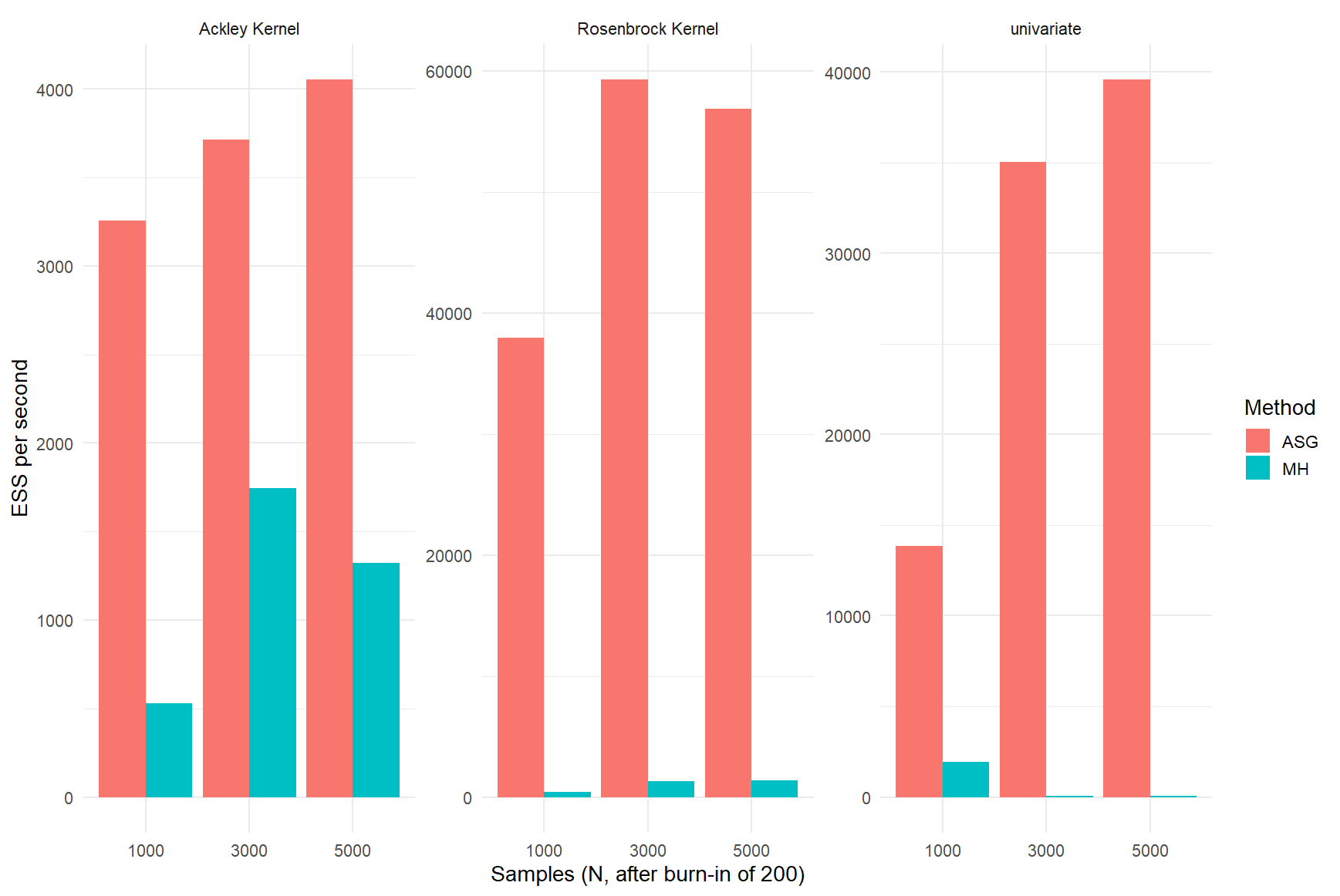}
    \caption{ ASG vs RW-MH for Multi ESS/s, fixed time}
    \label{fig:comp_uni_time}
\end{subfigure}
\hfill
\begin{subfigure}[b]{0.49\linewidth}
    \centering
    \includegraphics[width=\linewidth]{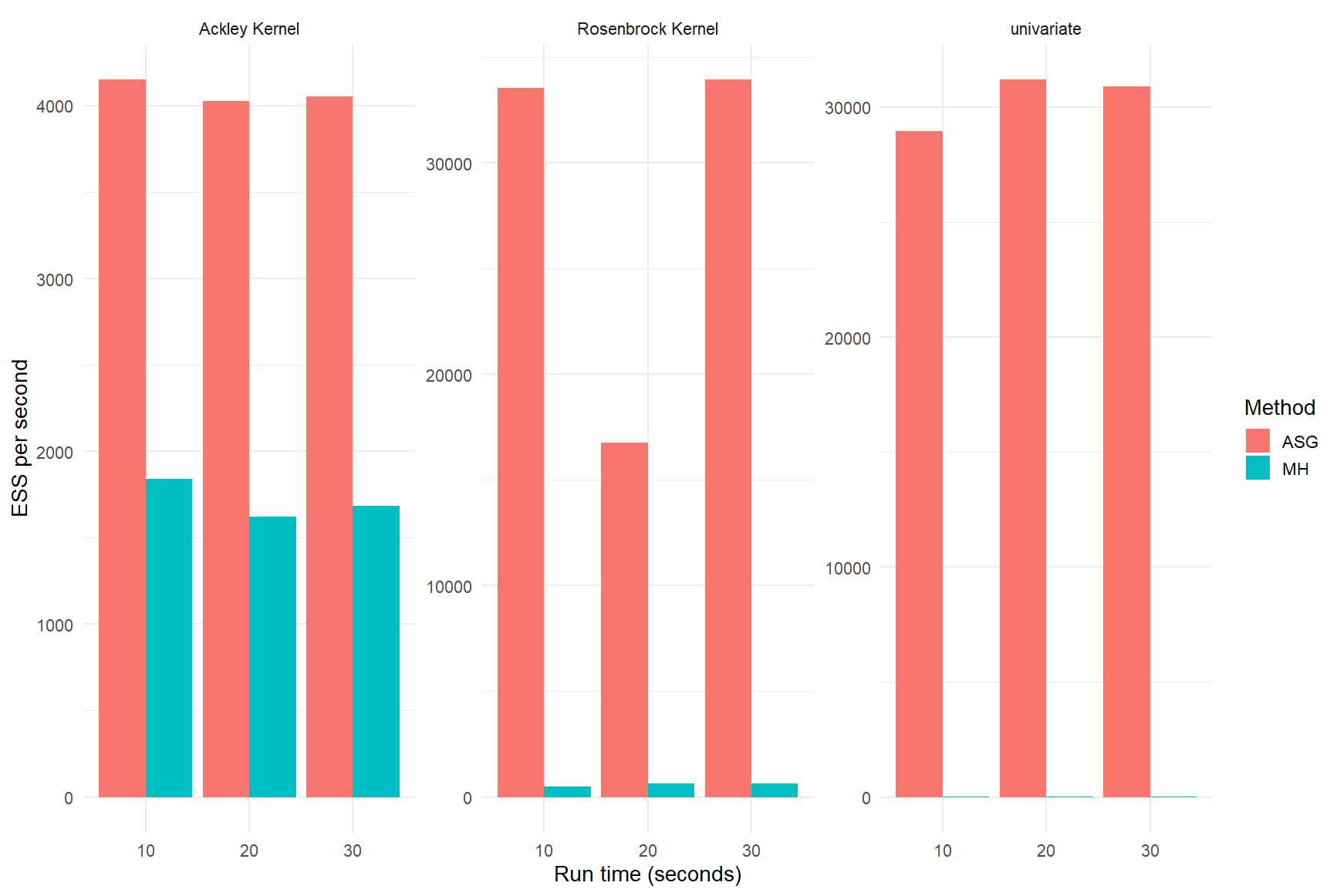}
    \caption{ ASG vs RW-MH for Multi ESS/s, fixed n}
    \label{fig:comp__sample}
\end{subfigure}
    \caption{Comparison figures of Algorithms in terms of Clock Time(here n = samples)}
\end{figure}

\subsection{Comparison of Algorithms in terms of Clock Time}
The performance of the proposed ASG Sampler and the RW-MH algorithm was evaluated based on the effective sample size per second  (Multi ESS/s), a metric that quantifies the efficiency of the sampling process by accounting for autocorrelation in the Markov chain. Figure~\ref{fig:comp_uni_time} illustrates the ESS per second (denoted Multi ESS/s) as a function of $\log_{10}(N)$, where $ N $ represents the total number of samples after burn-in, for both methods over a range of execution times (10, 20, 30s). The proposed method is shown in cyan, while the RW-MH method is shown in pink, with distinct markers corresponding to each time duration.

The results indicate that the proposed method consistently achieves higher Multi ESS/s values compared to the RW-MH algorithm, particularly at lower $\log_{10}(N)$ values, suggesting improved efficiency in the initial sampling phases. As $\log_{10}(N)$ increases the Multi ESS/s for both methods tends to decrease, reflecting increased autocorrelation or computational overhead with larger sample sizes. Notably, the proposed method maintains a relatively stable Multi ESS/s across all time durations, whereas the RW-MH method shows a more pronounced decline, especially at longer execution times. This suggests that the proposed approach may be more robust to scaling with sample size and time. A similar figure~\ref{fig:comp__sample} is shown where the Multi ESS/N with respect to time when the samples are fixed.

\section{Application of ASG to Statistical Loss Functions}\label{sec:non smooth}
The objective of applying the proposed ASG sampler to popular statistical loss functions such as LASSO and its generalizations is to demonstrate how optimization-based machine learning models can be recast as fully probabilistic inference problems and analyzed through sampling rather than point estimation alone. By transforming these loss functions into the corresponding probability kernels, we highlight how ASG can be used to generate draws from the entire posterior distribution, thereby enabling principled uncertainty quantification that is largely absent from standard ML practice, which typically relies only on mode estimation. This is particularly important for non-convex loss landscapes (equivalently, non-concave densities), where multiple local optima, sharp ridges, and irregular geometry make both optimization and sampling challenging. The aim of these examples is to show that ASG can efficiently explore such complex distributions, producing high-quality samples that capture multimodality and parameter uncertainty while remaining computationally practical. Performance will be compared with established MCMC methods, illustrating that ASG’s automated support calibration and slice-based updates provide substantial advantages for reliable and scalable Bayesian inference in statistical and machine learning models.

\subsection{Bridge loss function based kernel}\label{lasso}

We sample in $x=(\beta_0,\beta_1,\cdots,\beta_{20})$ from the unnormalized loss function density
\begin{equation}\label{eq:lasso/bridge}
K(x)\;\propto\;\exp\!\left\{-\left[
\frac{1}{2N}\sum_{i=1}^{N}\!\big(y_i-\beta_0-\mathbf z_i^\top\beta\big)^2
+\lambda\sum_{j=1}^{2}\lvert \beta_j\rvert\ ^{\alpha}\right]\right\}
\end{equation}
with $\beta=(\beta_1,\beta_2, \cdots, \beta_{20})^\top$, $\mathbf z_i$ taken as the standardized columns from \texttt{QuickStartExample} dataset, $N=\lvert y\rvert$, for simlicity first we took $\alpha = 1$ which led to LASSO, and $\lambda=0.1$. We used our proposed ASG sampler with
$n_{\text{samples}}=10,000$, burn-in $=2500$, thinning $=1$,
and starting value $x^{(0)}=(0,\cdots,0,0)$. 
Also, we computed the Multi ESS for each component of 
$x=(\beta_0,\beta_1,\cdots, \beta_{20})$ using the Multi ESS function. The resulting values were the Multi ESS score is 9245.
These large effective sample sizes (relative to the total of $10{,}000$ post--burn-in iterations)
indicate that the Markov chain achieved excellent mixing and that the posterior means are based on nearly independent draws from the target LASSO kernel. Fitting the same standardized design $Z$ and response $y$ with fixed
$\lambda=0.1$ gives the LASSO point estimate after using the glmnet package we get the results shown in the Table~\ref{tab:posterior_lasso_compare}. In addition, based on the samples generated by ASG, the ACF, the Trace and the running mean plots are given in the appendix.
\begin{table}[!ht]
\centering
\resizebox{\linewidth}{!}{%
\begin{tabular}{|c*{10}{c}|}
\toprule
& $\beta_0$ & $\beta_1$ & $\beta_2$ & $\beta_3$ & $\beta_4$ & $\beta_5$ & $\beta_6$ & $\beta_7$ & $\beta_8$ & $\beta_9$ \\
\midrule
Posterior Mode (ASG)
& 0.1109 & 1.3955 & 0.0297 & 0.7826 & 0.0526
& -0.9107 & 0.6458 & 0.1072 & 0.3946 & -0.0267 \\

ASG CI (2.5\%)
& -0.1019 & 1.1370 & -0.1789 & 0.5303 & -0.1408
& -1.1282 & 0.4342 & -0.0904 & 0.1784 & -0.2411 \\

ASG CI (97.5\%)
& 0.3524 & 1.4565 & 0.2273 & 0.9937 & 0.2620
& -0.6841 & 0.9001 & 0.3436 & 0.6001 & 0.1790 \\
\hline
\hline
Lasso Estimate
& 0.1500 & 1.3282 & 0 & 0.6735 & 0
& -0.8050 & 0.5269 & 0.0038 & 0.3254 & 0 \\
\bottomrule
\toprule
 $\beta_{10}$ & $\beta_{11}$ & $\beta_{12}$ & $\beta_{13}$ & $\beta_{14}$ & $\beta_{15}$ & $\beta_{16}$ & $\beta_{17}$ & $\beta_{18}$ & $\beta_{19}$ & $\beta_{20}$\\
\midrule
 0.1147 & 0.2517 & -0.0767 & -0.0170 & -1.1622
& -0.1592 & -0.0418 & -0.0650 & 0.0797 & -0.0005 & -1.1724 \\

-0.0801 & 0.0413 & -0.2988 & -0.2801 & -1.3811
& -0.4004 & -0.2842 & -0.2606 & -0.1359 & -0.2001 & -1.2087\\

 0.3498 & 0.5013 & 0.1534 & 0.1706 & -1.001
& 0.1148 & 0.1708 & 0.1586 & 0.2237 & 0.1901 & -1.003\\
\hline
\hline
 0 & 0.1393 & 0 & 0 & -1.0684
& 0 & 0 & 0 & 0 & 0 & -1.0030 \\
\midrule
\end{tabular}%
}
\caption{Comparison of Posterior Modes, ASG 95\% Credible Intervals, and LASSO Estimates}
\label{tab:posterior_lasso_compare}
\end{table}
\begin{figure}[!ht]
    \centering
    \includegraphics[width=\linewidth]{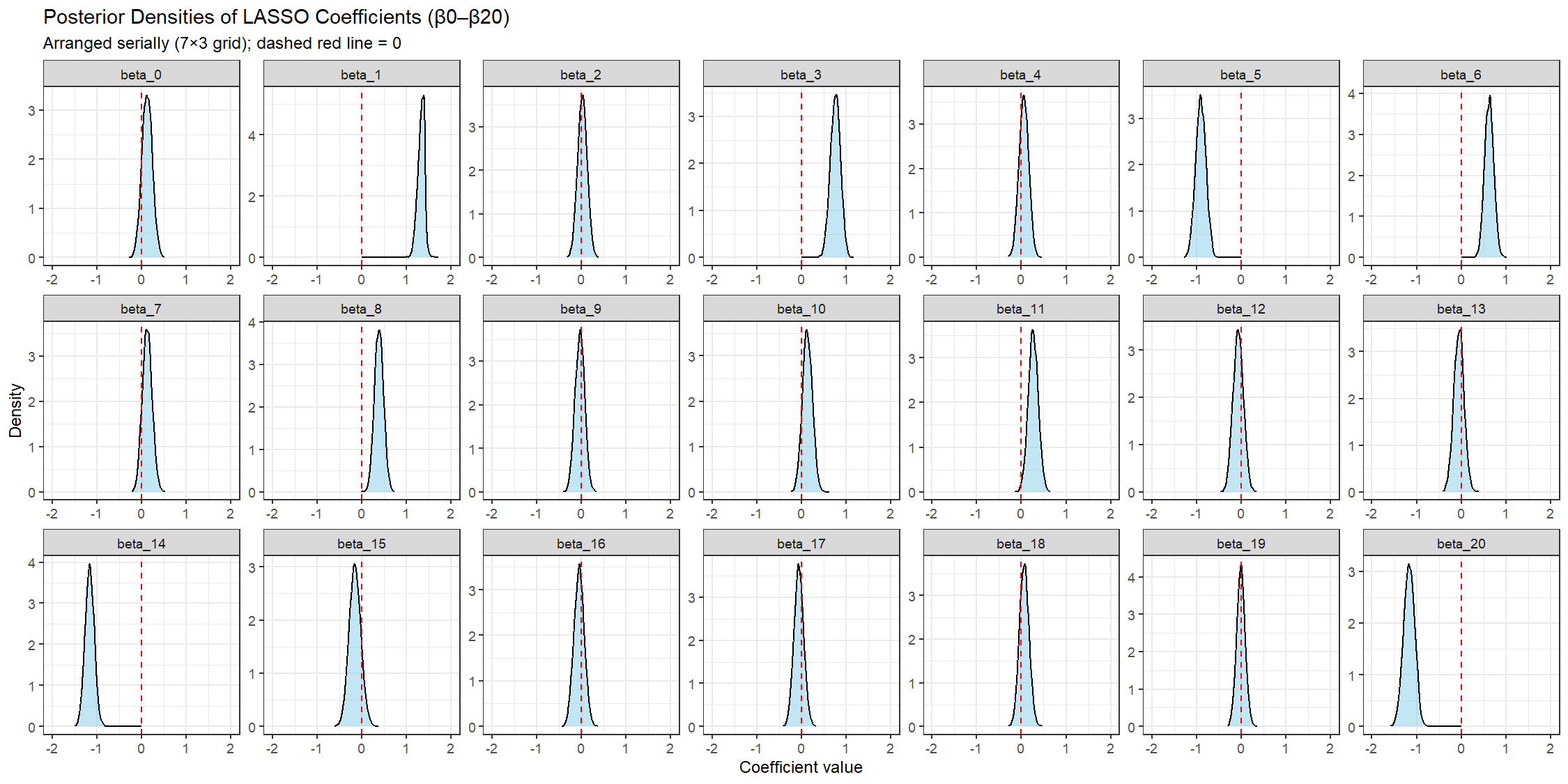}
    \caption{Posterior Densities of LASSO Coefficient after sampling from ASG Sampler.}
    \label{fig:placeholder}
\end{figure}
The \texttt{glmnet} solution minimizes the same objective and is MAP-like
point. 

Here another example of the Bridge Regression Kernel applied to the same dataset as mentioned in the section~\ref{lasso}.
We used the same setup for the tool $\lambda$ as 0.001 and $\alpha$ = 0.1 for the kernel~\eqref{eq:lasso/bridge}. The rest of the setup is the same as before. We had manually chosen the $\lambda$ value to show this example, as we all know that 
the $\lambda$ value decreases as we reduce the $\alpha$ values~\cite{Frank_1993,Fu_1998}.
The posterior modes are given in the table~\ref{tab:posterior_bridge_compare}.
The resulting values were Avg. Multi ESS = 9322.
These large ESS (relative to total of 10,000 post–burn-in)
indicate that the Markov chain achieved excellent mixing and that the posterior means are
based on nearly independent draws from the target Bridge regression kernel. Also we can see the posterior densities after sampling using the ASG sampler and how it looks like in the figure~\ref{fig:placeholder_br}. 
\begin{table}[!ht]
\centering
\resizebox{\linewidth}{!}{%
\begin{tabular}{|c*{10}{c}|}
\toprule
& $\beta_0$ & $\beta_1$ & $\beta_2$ & $\beta_3$ & $\beta_4$ & $\beta_5$ & $\beta_6$ & $\beta_7$ & $\beta_8$ & $\beta_9$ \\
\midrule
Posterior Mode (ASG)
& 0.1118 & 1.3799 & 0.0217 & 0.7822 & 0.0661
& -0.9285 & 0.6126 & 0.1256 & 0.3935 & -0.0423 \\

ASG CI (2.5\%)
& -0.1288 & 1.1685 & -0.1878 & 0.5442 & -0.1527
& -1.1337 & 0.4166 & -0.0955 & 0.1997 & -0.2576 \\

ASG CI (97.5\%)
& 0.3402 & 1.5934 & 0.2396 & 0.9862 & 0.2886
& -0.6744 & 0.8179 & 0.3430 & 0.6023 & 0.1798 \\
\hline
\hline
Lasso Estimate
& 0.1500 & 1.3282 & 0.0000 & 0.6735 & 0
& -0.8050 & 0.5269 & 0.0038 & 0.3254 & 0 \\
\bottomrule
\toprule
 $\beta_{10}$ & $\beta_{11}$ & $\beta_{12}$ & $\beta_{13}$ & $\beta_{14}$ & $\beta_{15}$ & $\beta_{16}$ & $\beta_{17}$ & $\beta_{18}$ & $\beta_{19}$ & $\beta_{20}$ \\
\midrule
 0.1436 & 0.2667 & -0.0719 & -0.0358 & -1.1693
& -0.1534 & -0.0365 & -0.0743 & 0.0594 & -0.0101 & -1.1485 \\

-0.0801 & 0.0271 & -0.2933 & -0.2753 & -1.3637
& -0.4054 & -0.2752 & -0.2675 & -0.1552 & -0.1933 & -1.3936 \\

 0.3579 & 0.4769 & 0.1546 & 0.1784 & -0.9614
& 0.1170 & 0.1700 & 0.1582 & 0.2703 & 0.1785 & -0.9018 \\
\hline
\hline
 0 & 0.1393 & 0 & 0 & -1.0684
& 0 & 0 & 0 & 0 & 0 & -1.0030\\
\bottomrule
\end{tabular}%
}
\caption{Posterior modes and 95\% credible intervals from ASG for the Bridge
regression kernel ($\alpha=0.1$, $\lambda=0.001$), with the LASSO
($\lambda=0.1$) estimates from \texttt{glmnet} shown for reference.}
\label{tab:posterior_bridge_compare}
\end{table}

\begin{figure}[!ht]
    \centering
    \includegraphics[width=\linewidth]{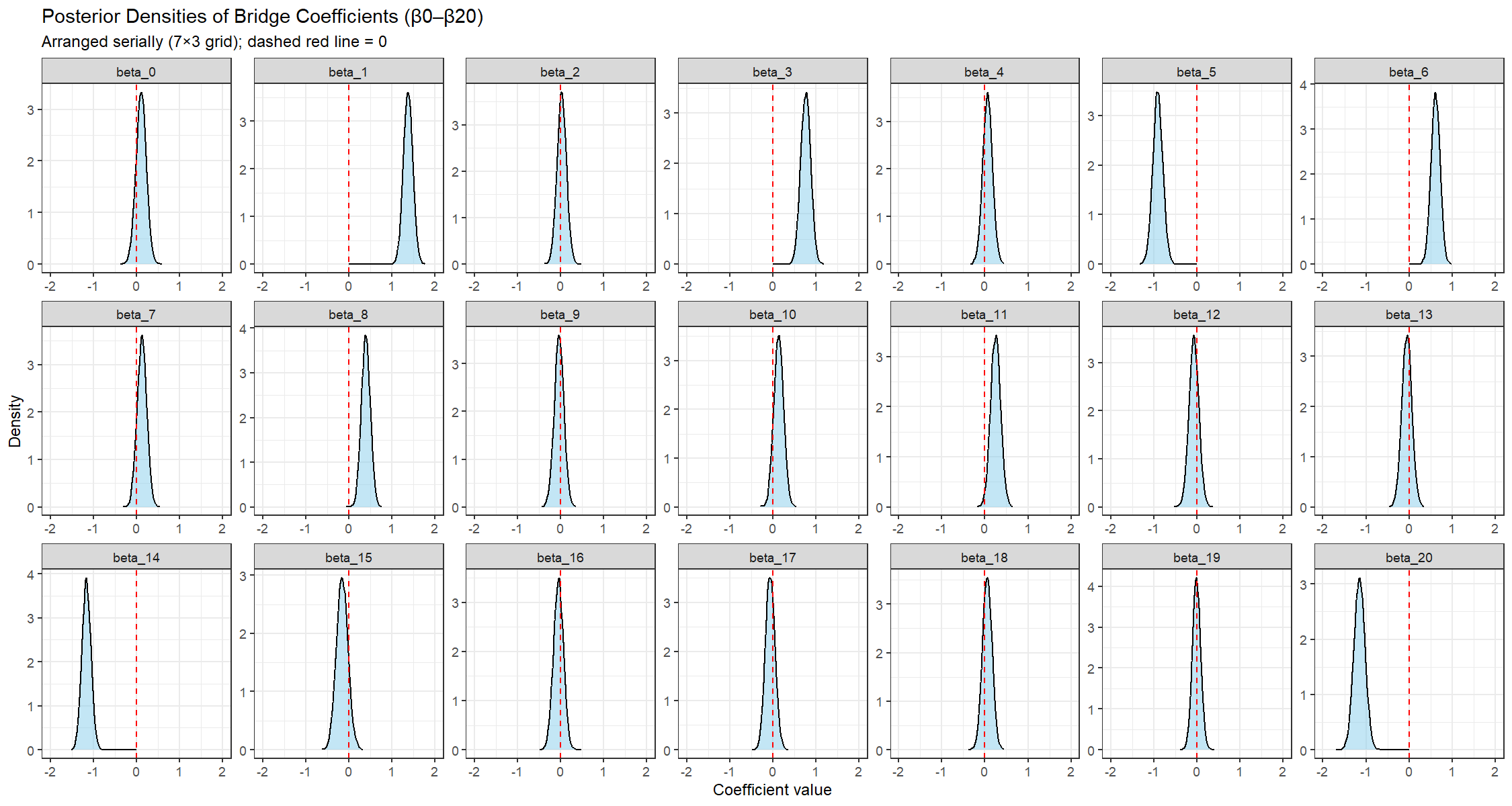}
    \caption{Posterior Densities of Bridge regression Coefficient sampled ASG Sampler.}
    \label{fig:placeholder_br}
\end{figure}

\section{Concluding Remarks and Future Directions}\label{sec:conclusion}\footnote{Results in this appendix reflect single runs and are included as
supplementary illustrations of ASG's behavior on additional
challenging kernels; they were not subjected to the replication
analysis of Section~\ref{sec:rosenbrock}.}
This paper introduced the ASG sampling framework, a fully automated MCMC method for sampling from arbitrary unnormalized probability kernels. The proposed approach combines Cauchy-based effective support estimation with slice-based coordinate updates within a Gibbs sampling structure. This construction preserves key theoretical properties such as stationarity, invariance, and ergodicity, while eliminating the need for manually specified proposal scales or support bounds.
Extensive numerical experiments demonstrate the robustness and efficiency of the ASG sampler across a diverse collection of challenging targets, including multimodal mixtures, highly curved benchmark functions such as the Rosenbrock and Ackley kernels, and non-smooth loss-based objectives such as LASSO. In all cases, the sampler produced stable chains with rapid mixing and very low autocorrelation. The resulting effective sample sizes were substantially larger than those obtained using the Random Walk Metropolis–Hastings algorithm, often by several orders of magnitude. Despite the additional computation required for automated support estimation, the ASG framework consistently achieved higher effective sample size per second, indicating superior time-normalized sampling efficiency.\\
These results highlight the advantages of combining adaptive effective-support determination with slice-based Gibbs updates for exploring complex probability landscapes. The empirical findings suggest that the ASG framework provides a practical and reliable sampling strategy for distributions that exhibit multimodality, sharp curvature, heavy tails, or non-smooth structure.\\
Several directions remain for future research. Theoretical analysis of spectral gap properties and convergence rates in higher-dimensional settings would further clarify the algorithm's asymptotic behavior. Extensions to block updates and parallel implementations could improve scalability for large-scale problems. Finally, integration of the ASG framework with modern probabilistic programming environments may broaden its applicability.

\section*{Code availability}
 Codes are made available upon request from the first author's GitHub page.\\ 
 \url{https://github.com/Prithwish-ghosh/AUTOMATED-SLICED-GIBBS-SAMPLER-FOR-ARBITRARY-KERNELS/tree/main} Github site
\bibliographystyle{plain}
\bibliography{sample}

\appendix\label{sec:appendix_examples}

\section{Additional Results and Plots}
\begin{proposition}[Validity of the bracketed rejection draw]
\label{prop:rejection}
Let $[a,b]$ be a bracket with $S_u \subseteq [a,b]$, and take as
proposal the uniform density on the bracket,
$  q(z) = \frac{\mathbb{I}(z \in [a,b])}{b-a}.$
Then $\mathbb{I}(z\in S_u) \le (b-a)\,q(z)$ for all $z$, and the
rejection scheme draw $z^\star \sim \mathrm{Unif}(a,b)$, accept $z^\star$ if $g(z^\star) > u$ (equivalently $z^\star \in S_u$),
returns, upon acceptance, an exact draw from $\mathrm{Unif}(S_u)$, with
$ P(\mathrm{accept}) = \frac{A(u)}{b-a}$ and hence
  $\mathbb{E}[\text{proposals until acceptance}] = \frac{b-a}{A(u)}.$
\end{proposition}
\begin{proof}
Since $S_u \subseteq [a,b]$, we have $\mathbb{I}(z\in S_u) \le
\mathbb{I}(z\in[a,b]) = (b-a)\,q(z)$ for all $z$, giving the envelope
condition with constant $(b-a)$. By the rejection sampling principle
above, accepted draws have density proportional to $\mathbb{I}(z\in
S_u)$ restricted to $[a,b]$, i.e.\ $\mathbb{I}(z\in S_u)/A(u)$, the
uniform density on $S_u$. The acceptance probability reduces to
$\int_{S_u} q(z)\,dz = A(u)/(b-a)$, and the number of proposals to
first acceptance is geometric with mean $(b-a)/A(u)$.
\end{proof}
Proposition~\ref{prop:rejection} shows 1d slice is
exact whenever $S_u \subseteq [a,b]$, with efficiency governed by the
coverage ratio $A(u)/(b-a)$. This ratio degrades when $S_u$ is a union
of narrow, well-separated components inside a wide bracket — the
disconnected/multimodal regime motivating ASG (Section~\ref{sec:intro});
in that case the expected number of proposals $(b-a)/A(u)$ grows
accordingly, which Remark~\ref{rem:bonferroni} shows compounds further
with dimension $m$.

\begin{remark}[Bonferroni correction for the multivariate tolerance]
\label{rem:bonferroni}
In Algorithm~\ref{alg:asg}, the effective-support step uses a common tolerance
$\epsilon_j=\epsilon/m$ across the $m$ coordinates. If $E_j$ denotes the event
that coordinate $j$'s slice extends beyond its computed bracket, then
$P(E_j)\le\epsilon_j$. By the union bound,
$P\left(\bigcup_{j=1}^m E_j\right)
\le \sum_{j=1}^m \epsilon_j=\epsilon$,
giving overall per-sweep coverage at least $1-\epsilon$. Without this
Bonferroni correction, using $\epsilon_j=\epsilon$ for every coordinate
yields the weaker bound $P(E)\le m\epsilon$, which deteriorates with
dimension. Equivalently, each coordinate uses per-side tail mass
$\epsilon/(2m)$ rather than $\epsilon/2$, producing wider brackets as $m$
increases.

%In Algorithm~\ref{alg:asg}, the effective-support step is invoked once per coordinate, using a common tolerance $\frac{\epsilon}{m}$. Let $E_j = \{x_j \notin [a_j,b_j]\}$ denote the event that coordinate $j$'s true conditional draw falls outside its computed bracket, so that by construction $P(E_j) \le \epsilon_j$. The event that the sweep fails to cover the slice in \emph{at least one} coordinate is $E = \bigcup_{j=1}^m E_j$, and by the union bound $  P(E) \;\le\; \sum_{j=1}^m \epsilon_j .$ If the uncorrected tolerance $\epsilon_j = \epsilon$ is used for every coordinate, this bound grows linearly with dimension, $  P(E) \;\le\; m\epsilon,$ which becomes vacuous once $m \ge 1/\epsilon$. To control the overall per-sweep miss probability at level $\epsilon$ independent of $m$, we apply a Bonferroni correction, setting $  \epsilon_j = \frac{\epsilon}{m}, \qquad j = 1,\dots,m,$ so that $  P(E) \;\le\; \sum_{j=1}^m \frac{\epsilon}{m} = \epsilon .$ Equivalently, the equal-tailed bracket for coordinate $j$ is computed with per-side tail mass $\epsilon/(2m)$ rather than $\epsilon/2$. This guarantees the same overall coverage $1-\epsilon$ for the full $m$-dimensional sweep at the cost of wider individual brackets $[a_j,b_j]$ as $m$ increases.
\end{remark}

\begin{remark}[Checking the box, in plain terms]
\label{rem:729-connection}
Before drawing, we verify that the box $[a,b]$ contains the entire slice above the threshold $u$. Specifically, we check the edge values $g_k(a)$ and $g_k(b)$. If both are below $u$, the box is valid; if either exceeds $u$, we reduce the tolerance, recompute the box, and repeat until both edges satisfy the check. This is the slice-sampling analogue of the classic accept--reject issue in Exercise 7.9~\cite{robert1999monte}: rather than assuming a bound is valid and correcting afterward, we verify it before sampling. Since the current point, threshold, and box change at every coordinate update, this verification is performed anew at each update.

%Before drawing a point, we check whether the box $[a,b]$ actually contains everything above the threshold $u$. The check is simple: look at the kernel's value at the two edges, $g_i(a)$ and $g_i(b)$. If both are already below $u$, the true slice cannot poke out past the box on either side, so the box is safe to use. If either edge value is still above $u$, the slice keeps going past that edge, and the box is too narrow. When that happens, we shrink the tolerance used to build the box and recompute it, repeating until both edges check out. This is the slice-sampling version of the classic accept--reject problem in Exercise 7.9~\cite{robert1999monte}: instead of assuming a bound is valid (and needing a correction step afterward if it turns out not to be), we verify it is valid before we ever draw a sample. Because the current point, the threshold $u$, and the box all change at every coordinate of every sweep, this check is repeated fresh each time, not just once at the start.
\end{remark}

\begin{figure}[!ht]
    \centering
     \begin{subfigure}[b]{0.45\linewidth}
        \includegraphics[width=\linewidth]{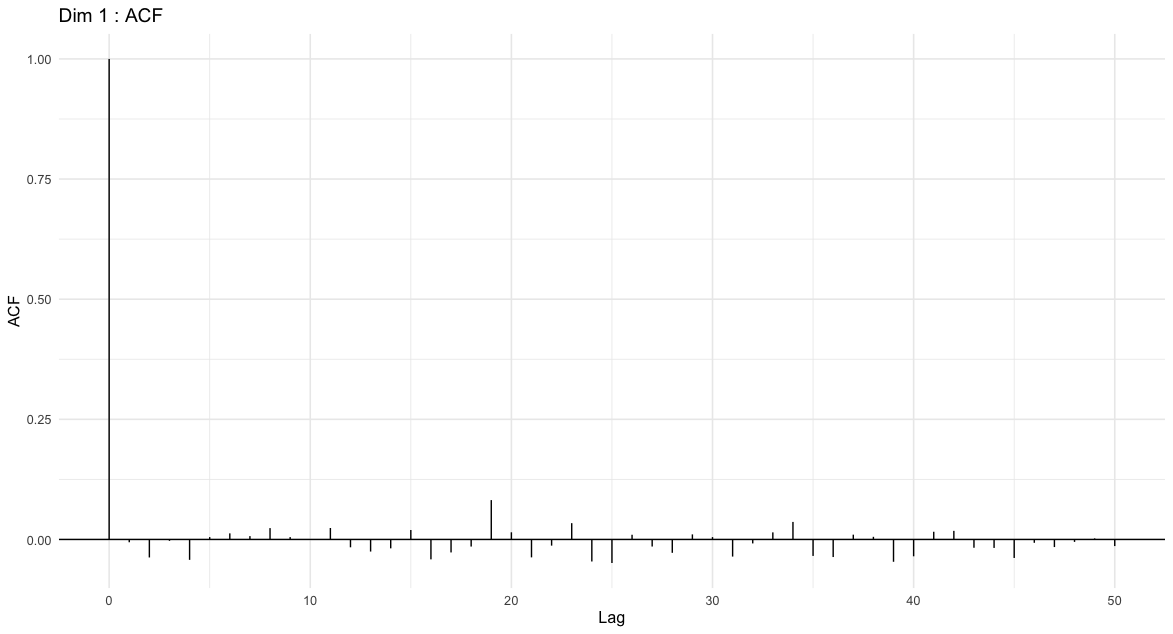}
        %\caption{ACF plot for ASG Sampler of posterior density}
    \end{subfigure}
    \hfill
     \begin{subfigure}[b]{0.45\linewidth}
        \includegraphics[width=\linewidth]{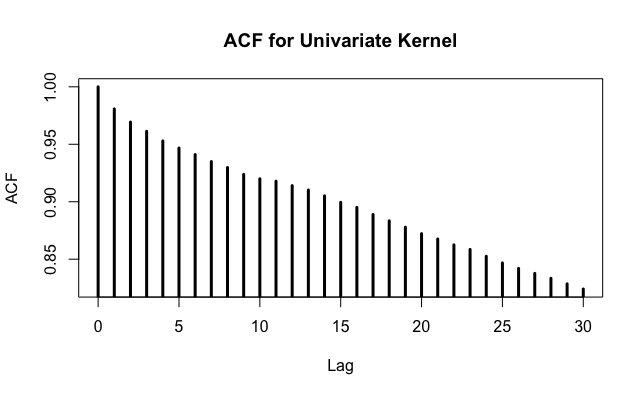}
        %\caption{ACF Plot for RW-MH algorithm of the posterior density}
    \end{subfigure}
    \hfill
     \begin{subfigure}[b]{0.45\linewidth}
        \includegraphics[width=\linewidth]{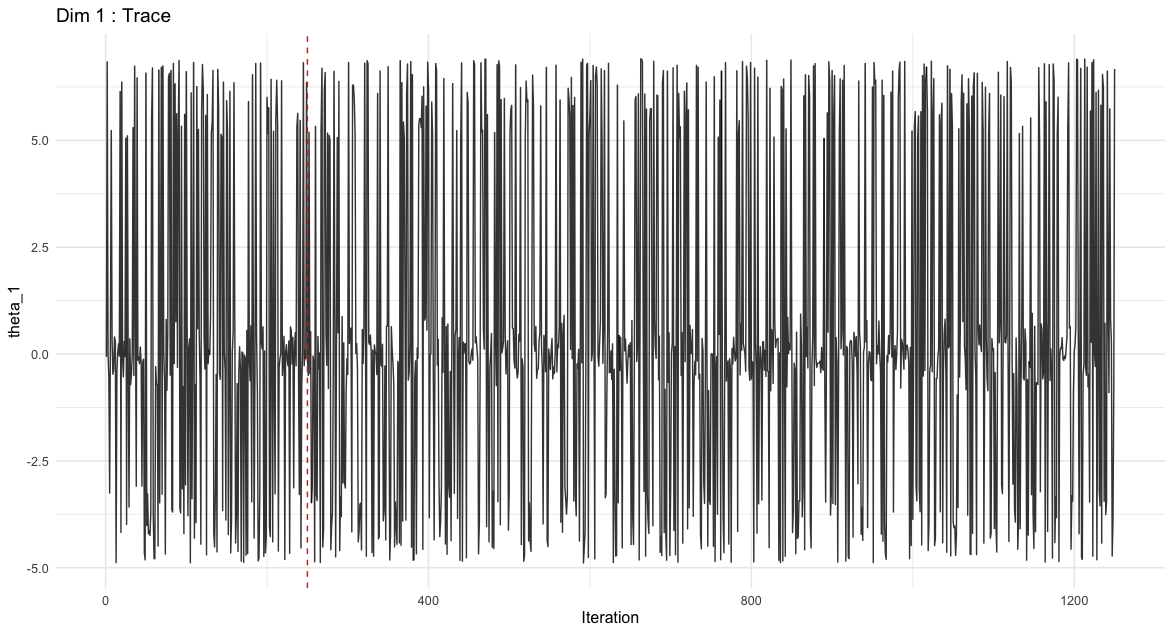}
        %\caption{Trace plot for ASG Sampler of the posterior density}
    \end{subfigure}
    \hfill
     \begin{subfigure}[b]{0.45\linewidth}
        \includegraphics[width=\linewidth]{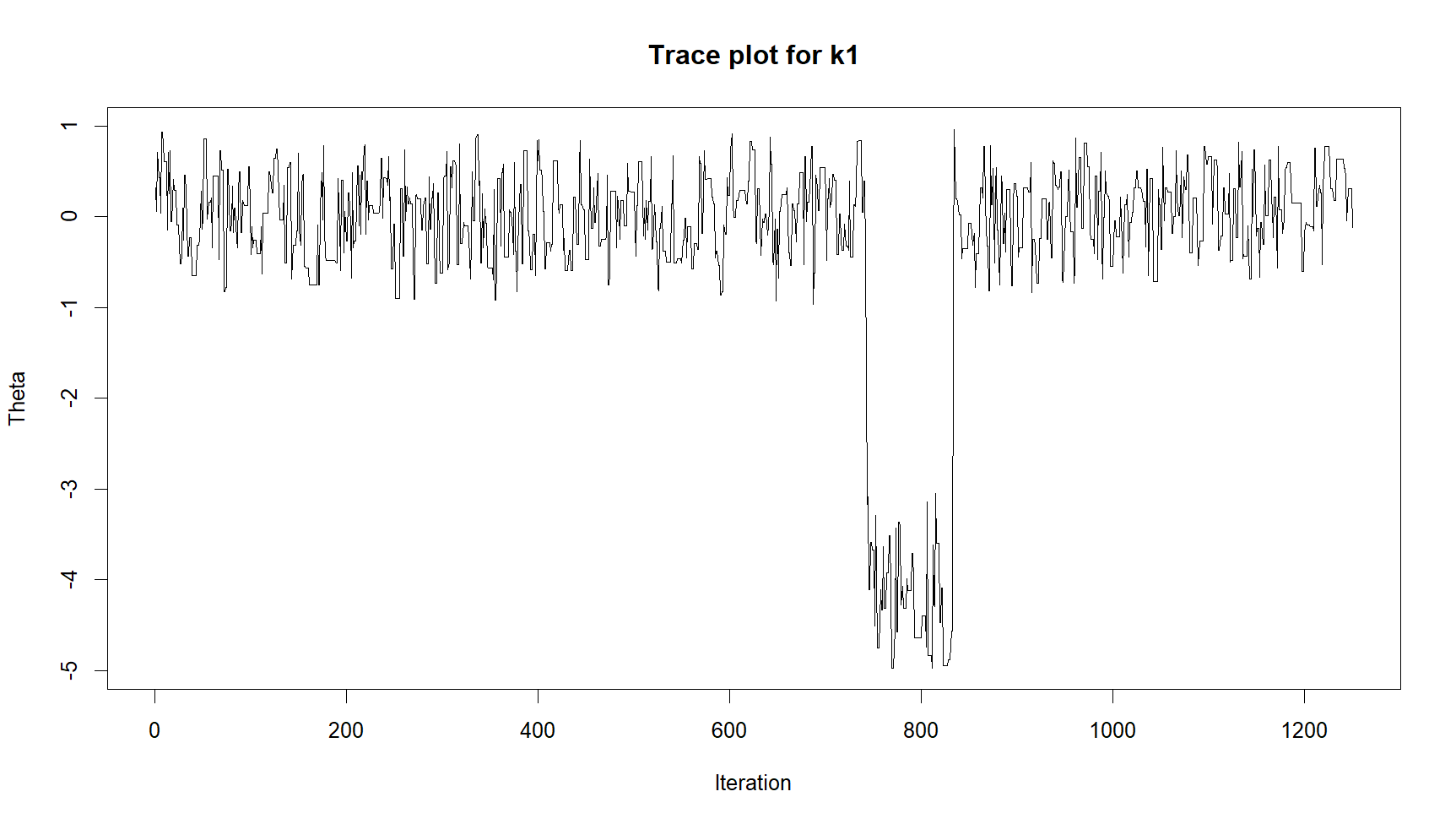}
        %\caption{Trace Plot for RW-MH algorithm of the posterior density}
    \end{subfigure}
\caption{Univariate Case Trace and ACF Plots of the posterior density}
\label{trace_uni}
\end{figure}

\begin{figure}[!ht]
    \centering
     \begin{subfigure}[b]{0.45\linewidth}
        \includegraphics[width=\linewidth]{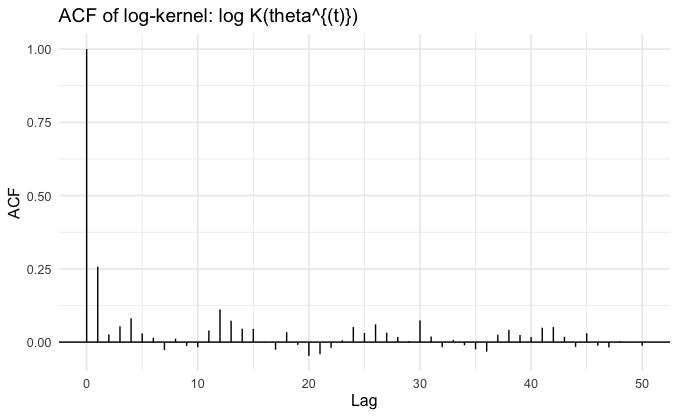}
        %\caption{ACF plot for ASG Sampler of the posterior density}
    \end{subfigure}
    \hfill
     \begin{subfigure}[b]{0.45\linewidth}
        \includegraphics[width=\linewidth]{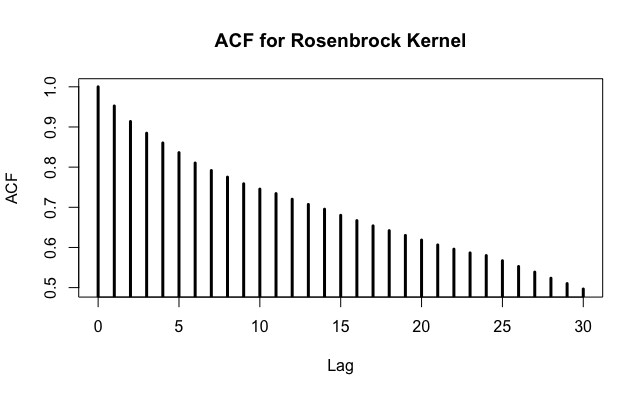}
       % \caption{ACF Plot for RW-MH algorithm of the posterior density}
    \end{subfigure}
\hfill
     \begin{subfigure}[b]{0.45\linewidth}
        \includegraphics[width=\linewidth]{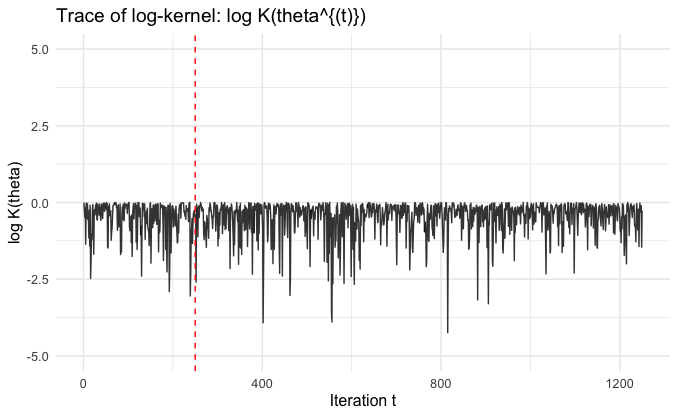}
       % \caption{Trace plot for ASG Sampler of the posterior density}
    \end{subfigure}
    \hfill
     \begin{subfigure}[b]{0.45\linewidth}
        \includegraphics[width=\linewidth]{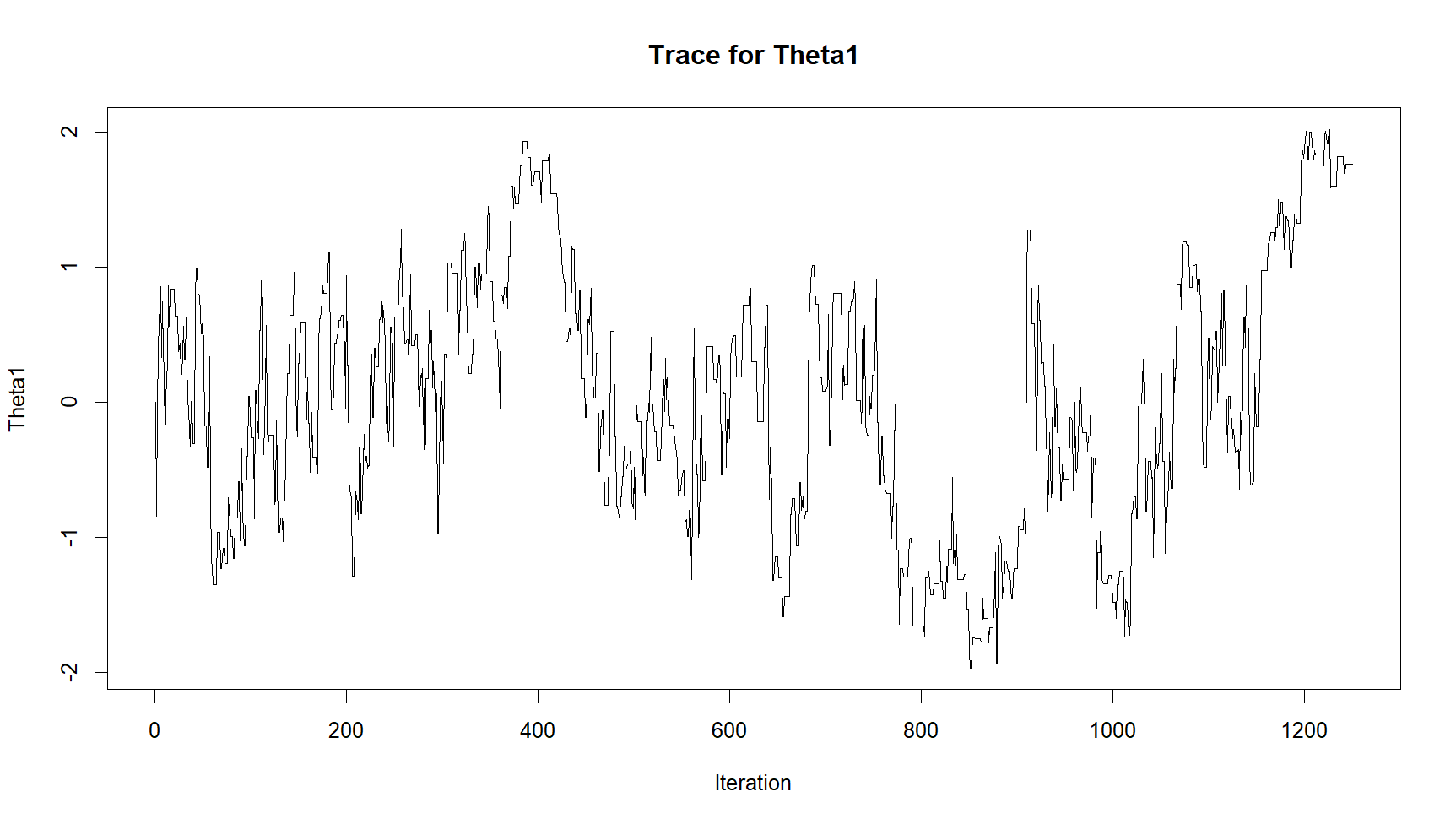}
      %  \caption{Trace Plot for RW-MH algorithm of the posterior density}
    \end{subfigure}
\caption{Rosenbrock (Banana) Shaped Case Trace and ACF Plots of the posterior density}
\label{trace_all_rose}
\end{figure}

\begin{figure}[!ht]
    \centering
     \begin{subfigure}[b]{0.45\linewidth}
        \includegraphics[width=\linewidth]{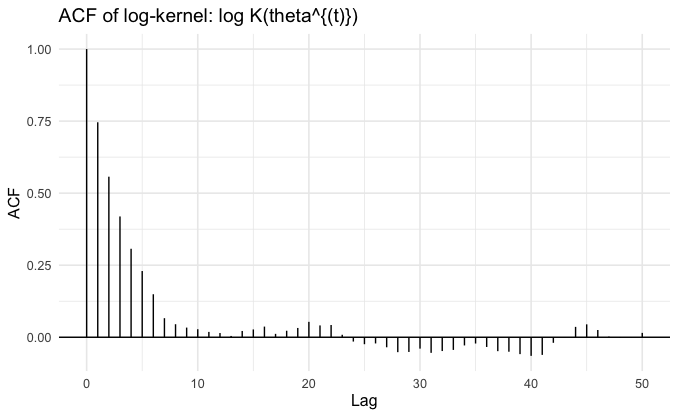}
        %\caption{ACF plot for ASG Sampler of the posterior density}
    \end{subfigure}
    \hfill
     \begin{subfigure}[b]{0.45\linewidth}
        \includegraphics[width=\linewidth]{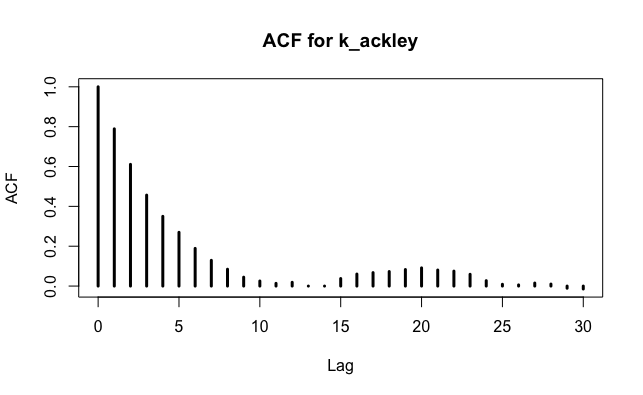}
        %\caption{ACF Plot for RW-MH algorithm of the posterior density}
    \end{subfigure}
\hfill
     \begin{subfigure}[b]{0.45\linewidth}
        \includegraphics[width=\linewidth]{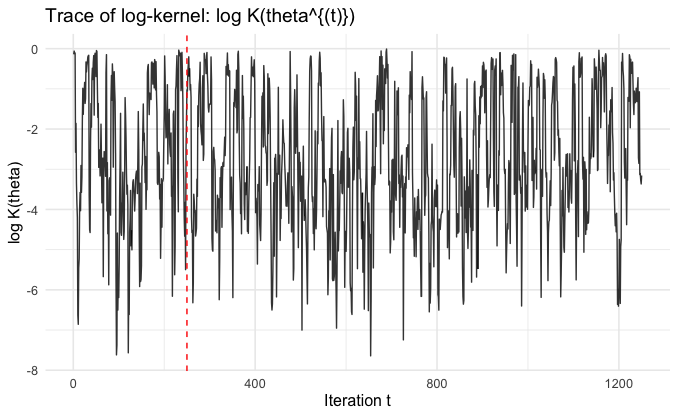}
        %\caption{Trace plot for ASG Sampler of the posterior density}
    \end{subfigure}
    \hfill
     \begin{subfigure}[b]{0.45\linewidth}
        \includegraphics[width=\linewidth]{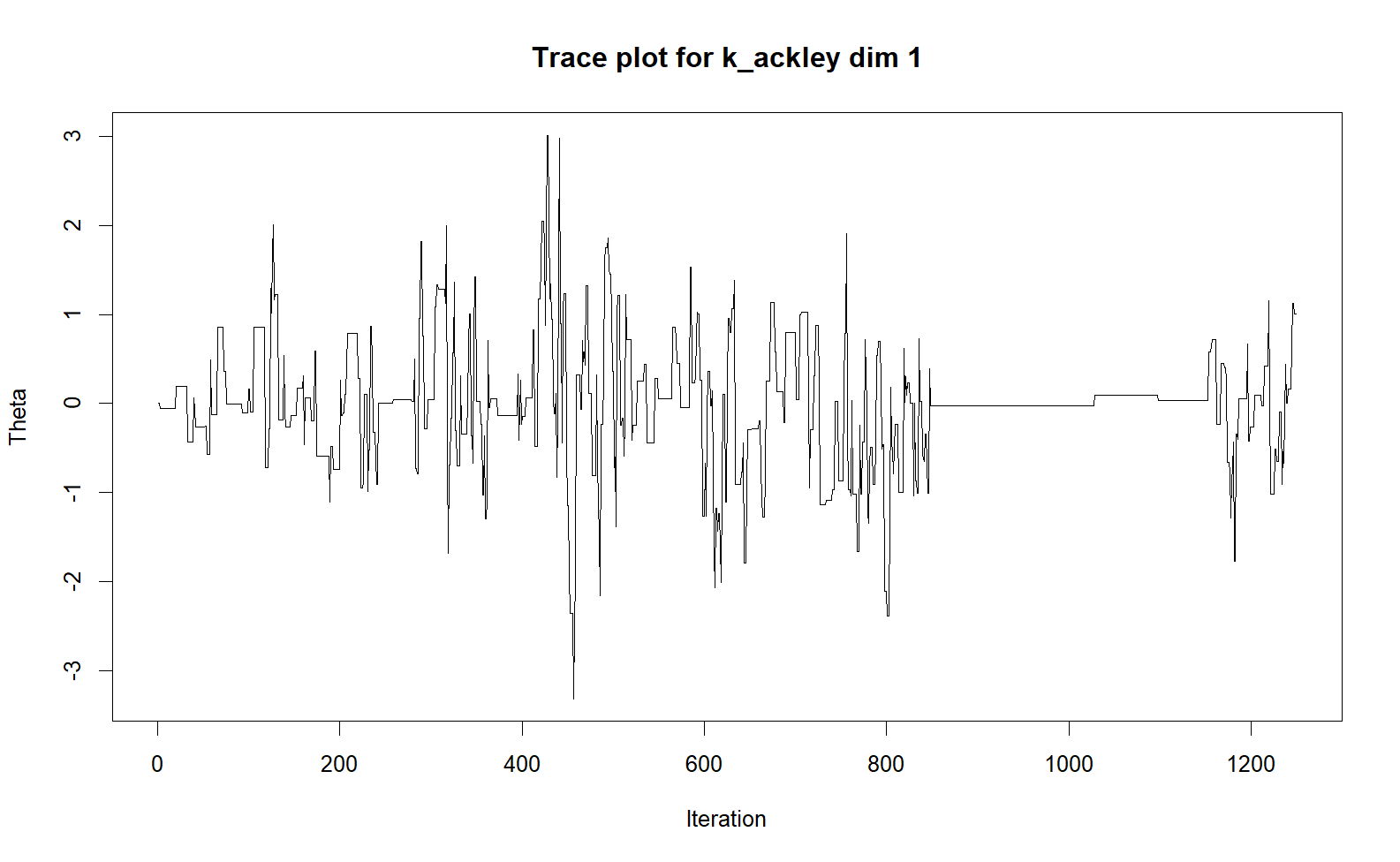}
       % \caption{Trace Plot for RW-MH algorithm of the posterior density}
    \end{subfigure}
\caption{Ackley Kernel Trace Plots of the posterior density}
\label{trace_all_ack}
\end{figure}
\begin{figure}[!ht]
    \centering
     \begin{subfigure}[b]{0.33\linewidth}
        \includegraphics[width=\linewidth]{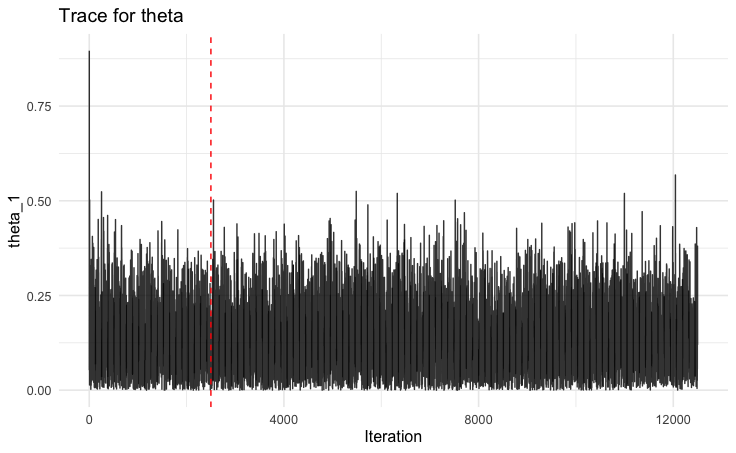}
        %\caption{Trace plot for ASG Sampler of the posterior density}
    \end{subfigure}
    \hfill
     \begin{subfigure}[b]{0.33\linewidth}
        \includegraphics[width=\linewidth]{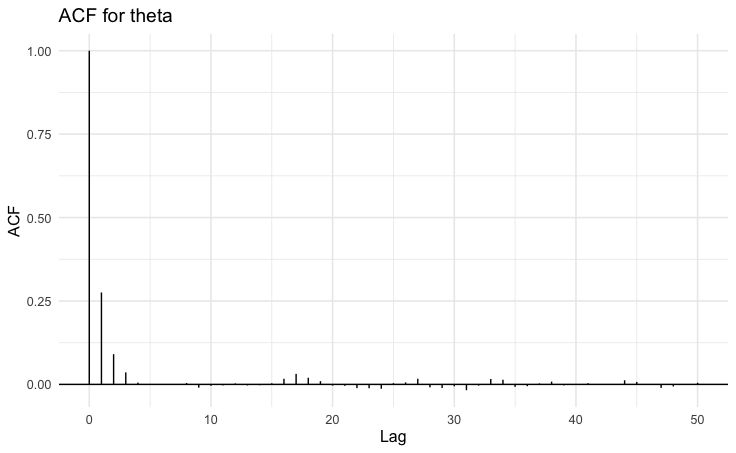}
        %\caption{ACF plot for ASG Sampler of the posterior density}
    \end{subfigure}
    \hfill
         \begin{subfigure}[b]{0.33\linewidth}
        \includegraphics[width=\linewidth]{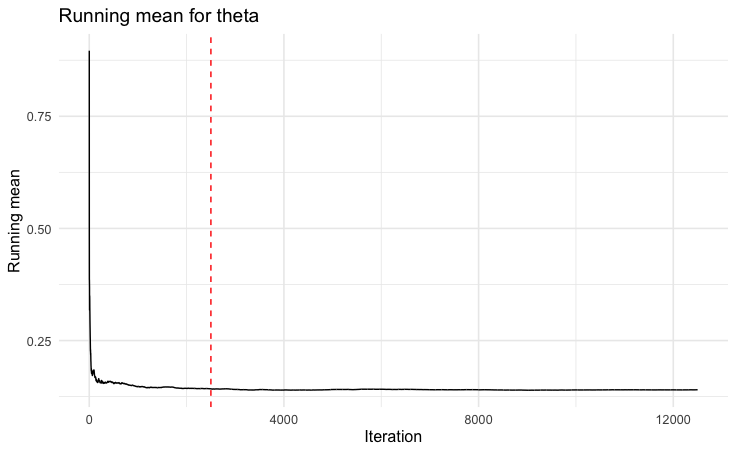}
       %\caption{Running Mean plot for ASG Sampler of the posterior density}
    \end{subfigure}
\caption{LASSO kernel Case Trace Plots of the posterior density}
\label{trace_all_lasso}
\end{figure}

\begin{figure}[!ht]
    \centering
     \begin{subfigure}[b]{0.33\linewidth}
        \includegraphics[width=\linewidth, height = 3.1cm]{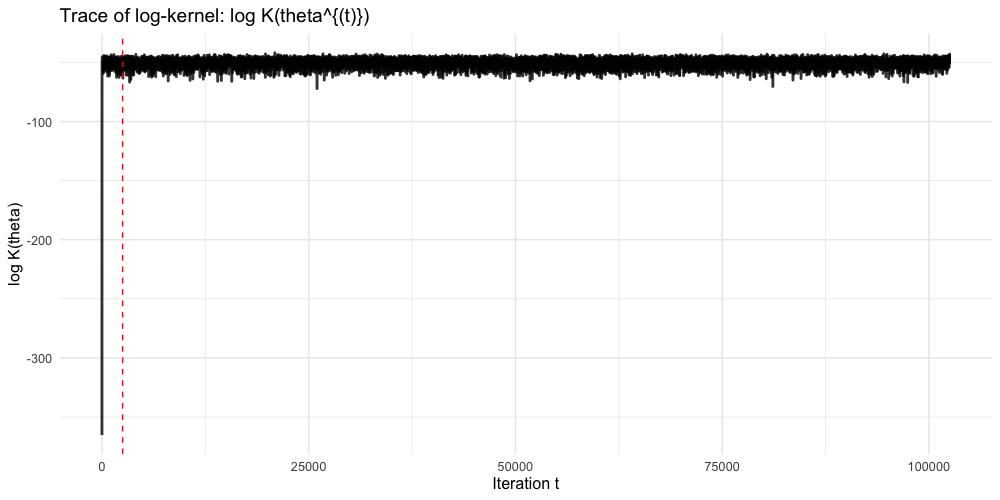}
       % \caption{Trace plot for ASG Sampler of the posterior density}
    \end{subfigure}
    \hfill
     \begin{subfigure}[b]{0.33\linewidth}
        \includegraphics[width=\linewidth]{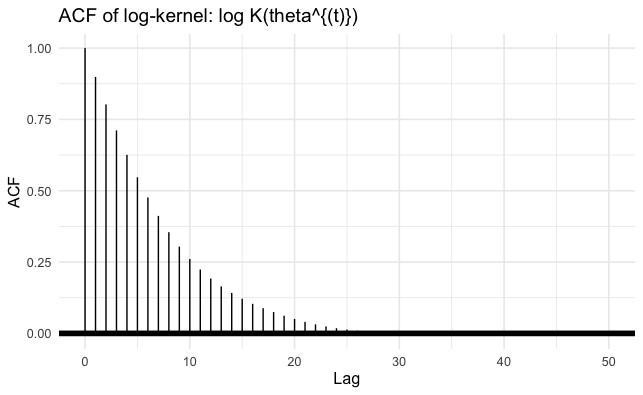}
       % \caption{ACF plot for ASG Sampler of the posterior density}
    \end{subfigure}
    \hfill
         \begin{subfigure}[b]{0.33\linewidth}
        \includegraphics[width=\linewidth]{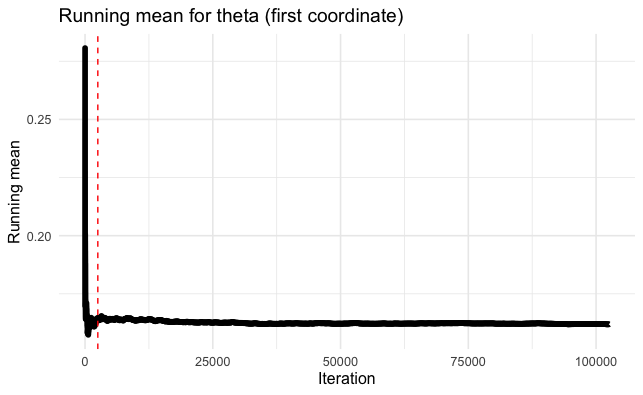}
       % \caption{Running Mean plot for ASG Sampler of the posterior density}
    \end{subfigure}
\caption{Bridge Regression kernel Case Trace Plots of the posterior density}
\label{trace_all_bridge}
\end{figure} 

\begin{figure}[!ht]
\centering
\begin{subfigure}[b]{0.33\linewidth}
    \centering
    \begin{tabular}{l|rr}
    \textbf{Sampler} & \textbf{Time} & \textbf{ESS}  \\
    \hline
    \hline
      \textbf{Proposed ASG}   & 0.71  &  1853    \\
      \textbf{RW-MH}   & 0.26  & 16  \\
      \textbf{HMC} & 1.20 & 15 \\
      \textbf{Adaptive Gibbs} & 0.04 & 19 \\
      \textbf{Qslice} & 0.40  & 888 \\
      \textbf{Elliptical slice} & 0.05 & 186  \\
    \end{tabular}
    \caption{Comparison with ASG }
    \label{tab:comparison_funnel}
\end{subfigure}
\hfill
\begin{subfigure}[b]{0.33\linewidth}
    \includegraphics[width=0.8\linewidth]{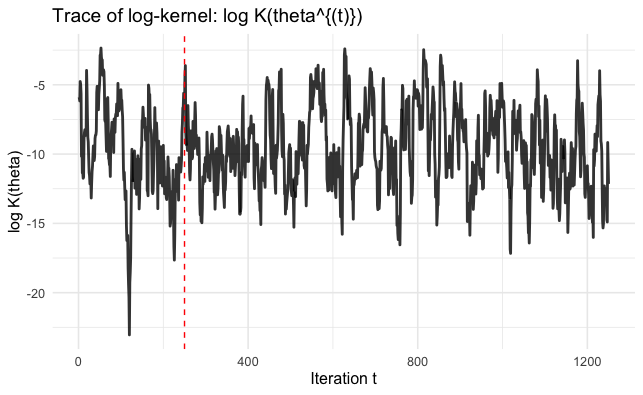}
    \caption{Trace plot for ASG }
\end{subfigure}
\hfill
\begin{subfigure}[b]{0.33\linewidth}
    \includegraphics[width=0.8\linewidth]{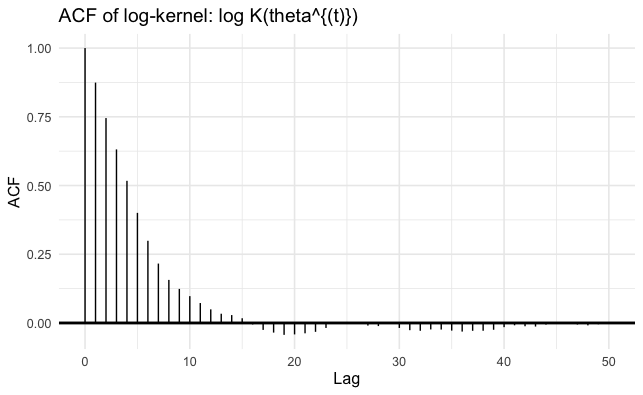}
    \caption{ACF plot for ASG }
\end{subfigure}
\caption{Funnel Distributional kernel Case Trace, ACF Plots of the posterior density With ESS}
\label{trace_all_funn}
\end{figure}

\begin{figure}[!ht]
\centering
\begin{subfigure}[b]{0.33\linewidth}
    \centering
    \begin{tabular}{l|rr}
    \textbf{Sampler} &  \textbf{Avg ESS}  \\
    \hline
    \hline
      \textbf{Proposed ASG}   &  492    \\
      \textbf{RW-MH}   &  16  \\
      \textbf{HMC} &  10 \\
      \textbf{Adaptive Gibbs} & 78 \\
      \textbf{Qslice} &  402 \\
          \textbf{Elliptical slice} &   404 \\
    \end{tabular}
    \caption{Comparison with ASG}
    \label{tab:comparison_sqdk}
\end{subfigure}
\hfill
     \begin{subfigure}[b]{0.33\linewidth}
        \includegraphics[width=0.8\linewidth]{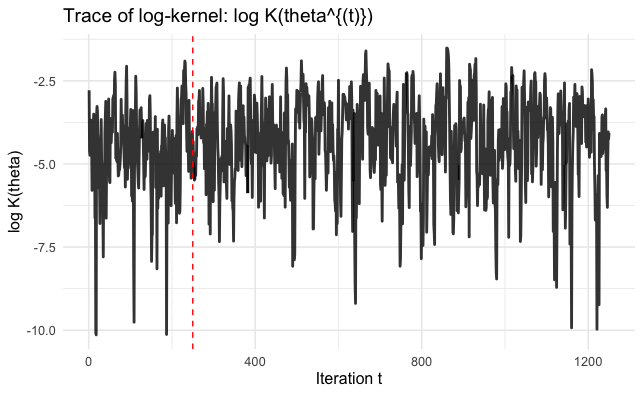}
        \caption{Trace plot of ASG }
    \end{subfigure}
    \hfill
     \begin{subfigure}[b]{0.33\linewidth}
        \includegraphics[width=0.8\linewidth]{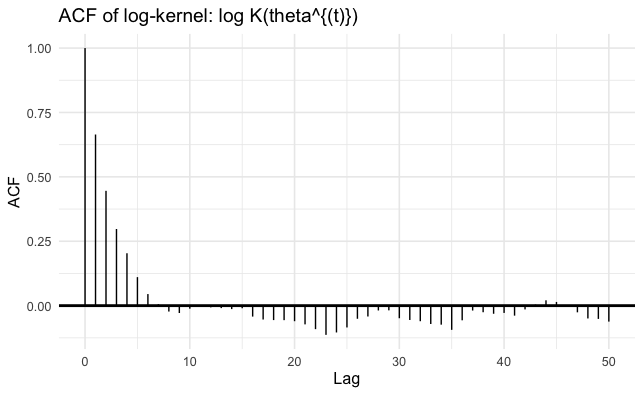}
        \caption{ACF plot for  ASG }
    \end{subfigure}
\caption{Allen cahn kernel Case Trace, ACF Plots of the posterior density, with ESS}
\label{trace_all_sqig}
\end{figure}
\begin{figure}[!ht]
    \centering
     \begin{subfigure}[b]{0.4\linewidth}
        \includegraphics[width=\linewidth]{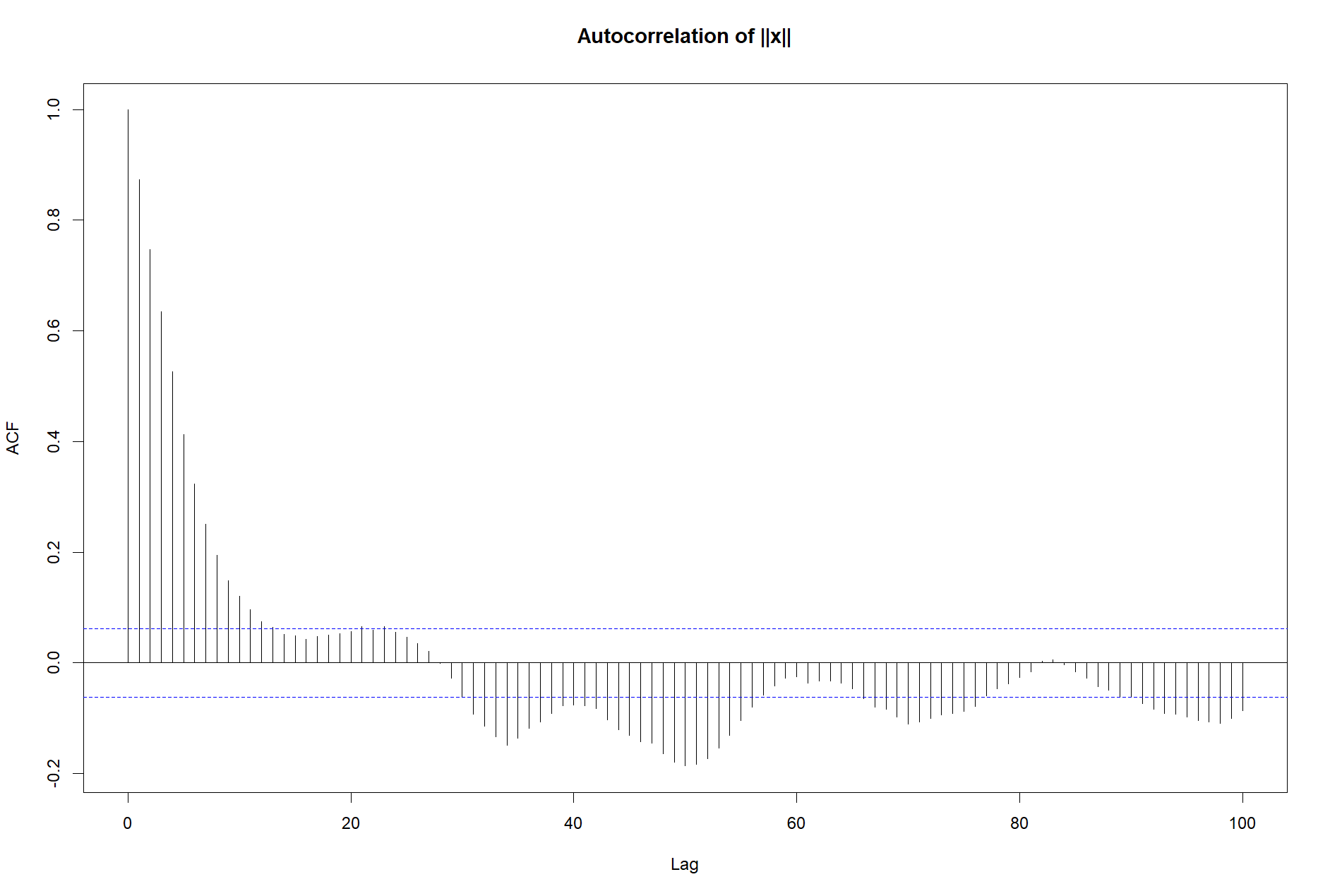}
    \end{subfigure}
     \begin{subfigure}[b]{0.4\linewidth}
        \includegraphics[width=\linewidth, height = 3.5cm]{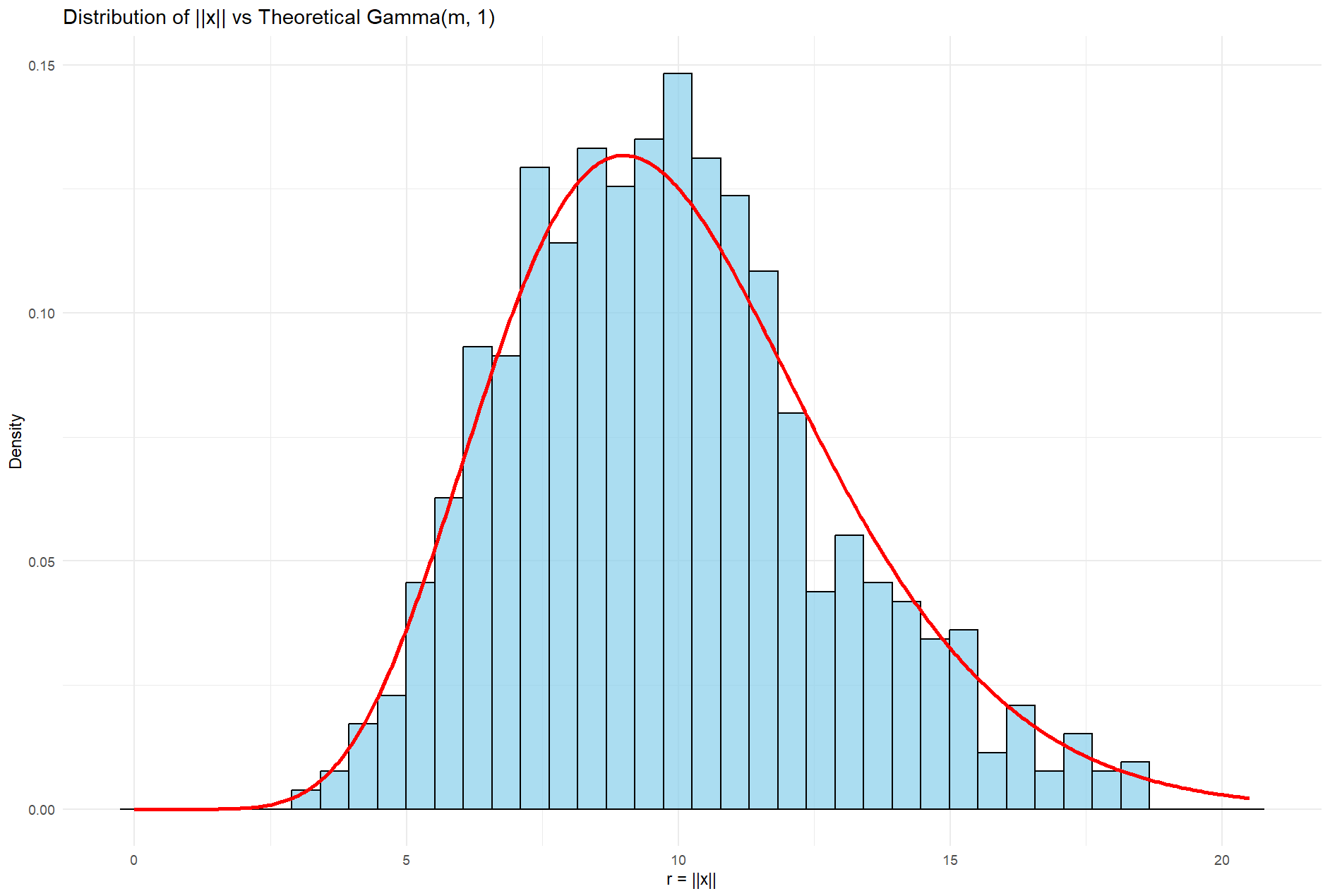}
    \end{subfigure}
    \caption{Performance of the ASG for the Euclidean norm kernel, (a)ACF of the sampled radii $\|x\|$ (b) compares the empirical density of samples generated by AS}
    \label{radii_comp}
\end{figure}

\end{document}